\documentclass[ejsv2]{imsart}

\RequirePackage[numbers]{natbib}
\RequirePackage[colorlinks,citecolor=blue,urlcolor=blue]{hyperref}
\RequirePackage{graphicx}

\usepackage{bm}
\usepackage{graphicx}
\def\bSig\mathbf{\Sigma}

\usepackage{amsmath}
\usepackage{color}
\usepackage{algorithm}
\usepackage{algpseudocode}
\usepackage{algorithmicx}
\usepackage{multirow}
\usepackage{mathtools}
\usepackage{hyperref}
\usepackage{natbib}
\usepackage{makecell}

\usepackage{url}

\startlocaldefs
\theoremstyle{plain}

\newtheorem{theorem}{Theorem}[section]
\newtheorem{lemma}[theorem]{Lemma}
\newtheorem{assumption}{Assumption}
\newtheorem{proposition}{Proposition}
\theoremstyle{definition}

\theoremstyle{remark}

\endlocaldefs

\begin{document}
\begin{frontmatter}
\title{A Fast Screening Approach for High-dimensional Outcomes and High-dimensional Predictors}
\runtitle{GIDS}

\begin{aug}
\author[A]{\fnms{Hongju}~\snm{Park}\ead[label=e1]{hongju@unc.edu}},
\author[B]{\fnms{Zhenyao}~\snm{Ye}\ead[label=e2]{zye@som.umaryland.edu}}
\and
\author[B]{\fnms{Shuo}~\snm{Chen}\ead[label=e3]{shuochen@som.umaryland.edu}}
\address[A]{School of Pharmacy,
University of North Carolina, Chapel Hill\printead[presep={,\ }]{e1}}

\address[B]{Maryland Psychiatric Research Center, School of Medicine,
University of Maryland\printead[presep={,\ }]{e2,e3}}
\runauthor{Hongju Park et al.}
\end{aug}

\begin{abstract}
Modeling interactions among multimodal, high-dimensional data is intrinsically challenging due to ultra-high dimensionality and complex dependence structure with high level noise. Screening methods are effective for reducing dimensionality, but most existing approaches shrink only the predictor space while retaining all outcomes. In cross-modal analyses, different outcomes often select different predictor subsets, so the union remains large and the response dimension is unchanged—limiting the practical benefit of screening. This gives rise to heavy computational burdens and poor interpretability. To address these limitations, we propose a new screening framework, Graph Independence Dual Screening (GIDS), which simultaneously reduces the dimensionality of response variables and predictors. We design computationally efficient algorithms that facilitate downstream selection procedures, improving accuracy and scalability, and establish supporting theoretical results. Exensive simulation studies demonstrate that GIDS outperforms the existing methods, which screen only predictors. To illustrate its utility, we applied GIDS to the Alzheimer’s Disease Neuroimaging Initiative (ADNI) dataset, analyzing interactions between genome-wide 865,353 DNA methylation and 49,386 transcriptomic variables. GIDS reduced the feature space to approximately 9,000 CpGs and 2,000 transcripts, uncovering blockwise interaction structures—clusters of CpG sites and gene transcripts with strong associations. These findings not only improve computational tractability but also yield interpretable biological insights, highlighting coordinated regulatory mechanisms underlying Alzheimer’s disease.
\end{abstract}



\end{frontmatter}

\section{Introduction}


The joint analysis of multi-modal high-dimesional data has garnered increasing attention as a powerful approach for uncovering relationships between two high-dimensional datasets. For example, multi-modal high-dimesional data studies for the multi-omics such as genome-wide single-nucleotide polymorphism (SNP) and transcriptomics \citep{lock2013joint,zhou2024fair,lee2024dcca} and imaging-genetics \citep{zhao2021transcriptome,zhao2021large} often face ultra-high dimensional and complex interactive patterns between datasets. These challenges lead to substantial computational burdens, such as the memory demands required to store a cross-correlation matrix involving both high-dimensional outcomes and predictors \citep{becker2023large}. For instance, the joint analysis of DNA methylation and transcriptomics data considered in this work involves $5\times 10^4$ and $8\times 10^5$ variables on each side, resulting in a cross-correlation matrix with roughly $4\times 10^{10}$ entries, which requires approximately 300GB memory.


In addition, in ultra-high dimensional setups, a major technical challenge is controlling the noise introduced by spurious correlations, as the magnitude of spurious correlations significantly grows uncontrollably with the number of predictors \cite{fan2008sure}. This issue is further exacerbated in joint setups, where the effective dimensionality increases multiplicatively with the numbers of predictors and responses. Consequently, large spurious correlations between unrelated pairs become much more frequent, which can severely bias variable ranking. Existing methods for multiple responses \citep{li2012feature,shao2014martingale,pan2019generic,liu2022model}, when screening predictors only, often mitigate this problem using hard or soft thresholding for screening measures of predictors. However, these techniques are not well-suited for joint screening, as they fail to account for the noise caused by ultra-high dimensional responses.


We present this challenge using the following notations. For two high-dimensional data sets $\mathbf{X} = [X_1~X_2~\cdots~X_p] \in \mathbb{R}^{n \times p}$ and $\mathbf{Y}  = [Y_1~Y_2~\cdots~Y_q] \in \mathbb{R}^{n \times q}$ with $n$ samples, then an assosiation matrix lies in $\mathbb{R}^{p \times q}$. As both $p$ and $q$ are high-dimensional, reducing $p$ or $q$ separately cannot solve this challenges, as mentioned above. Our goal is to seek a efficent solution to simultaneously reduce $p$ and $q$ while preserving most true associations between $\mathbf{X}$ and $\mathbf{Y}$. To this end, we consider the problem of estimating a $p\times q$-dimensional parameter matrix $\bm{\beta} \in \mathbb{R}^{p \times q}$ in the multivariate linear model  
\begin{eqnarray*}
    \mathbf{Y} = \mathbf{X}\bm{\beta} + \bm{\epsilon},
\end{eqnarray*}
where $\mathbf{Y}$ is a response matrix, $\mathbf{X}$ is a design matrix with independent and identically distributed (i.i.d.) rows, $\bm{\beta} \in \mathbb{R}^{p \times q}$ is the unknown coefficient matrix, $\bm{\epsilon}\in \mathbb{R}^{p \times q}$ is a random error matrix, and $n$ is the sample size. In high-dimensional settings, where both $p$ and $q$ may be large, it is commonly assumed that only a small subset of predictors and responses are truly associated. Under the joint high-dimensional setup, variable selection plays a crucial role in improving estimation accuracy by isolating the relevant predictors. Additionally, it enhances interpretability by yielding a parsimonious representation of the underlying model structure.

Screening has long played a central role in high-dimensional data analysis to reduce the dimension to the managible size. Especially, correlation-based screening methods have widely used due to their simplicity, interpretability, and scalability. A notable example is the Sure Independence Screening (SIS) method introduced by \citep{fan2008sure}, which uses marginal Pearson correlations to efficiently reduce dimensionality by filtering out irrelevant predictors. However, such methods are primarily designed for single-response models and may fail to capture complex dependence structures inherent in high-dimensional multivariate settings. In practice, different outcomes often correspond to distinct subsets of predictors, so their union remains large and the response dimension is unchanged, resulting in limited dimensionality reduction and reduced interpretability.
Furthermore, this approach lacks a mechanism for screening responses, highlighting the need for joint screening of both predictors and responses. The following proposition illustrates this limitation.

\begin{proposition}
Let $\bm{\beta} \in \mathbb{R}^{q \times p}$ with true model $\mathcal{M} := \{(i,j): \beta_{ij}\neq 0\} = \{(i,j): 1\leq i \leq p_0,~1\leq j \leq q_0\}$ and true model for the $j$th outcome $\mathcal{M}_j = \{i : (i,j)\in \mathcal{M}\}$. Suppose that an independent screening is performed for each outcome with the true and false positive probablity of selecting true predictors such that
$$
\mathbb{P}\big(i\in \widehat{\mathcal{M}}_j \,\big|\, i\in \mathcal{M}_j\big)=1-\xi_1, 
\quad \text{and} \quad 
\mathbb{P}\big(i\in \widehat{\mathcal{M}}_j \,\big|\, i\notin \mathcal{M}_j\big)=\xi_2,
$$
where $1-\xi_1 \gg \xi_2$. Then, the probability of union of screening results for $q$ outcomes satisfies 
$$
\mathbb{P}\Big( \big|\cup_{j=1}^q \widehat{\mathcal{M}}_j\big| = p \Big) \;\geq\; 1 - p(1-\xi_2)^q .
$$
\label{prop:1}
\end{proposition}

This proposition guarantees a lower bound on the probability that no predictors are excluded with high probability, when performing screening for $q$ outcomes separately. For example, when $p=q=10,000$, the minimum probability of screening results including all predictors is 1.000 for $\xi_2= 0.005$. Therefore, traditional screening methods may not exclude any predictors and outcomes for the downstreaming analyses, which poses both computational challenges and replicablity crisis. 

Over the past decade, many existing works have extended single-response screening to the joint analysis of predictors and responses. For example, SIS has been extended to handle multiple responses using the distance correlation \citep{li2012feature}, the martingale difference correlation \citep{shao2014martingale}, and the ball correlation \citep{pan2019generic}. Furthermore, screening methods using a rank-based correlation \citep{he2021sure} and projection correlation \citep{liu2022model} have been proposed for joint high-dimensional settings. However, these methods still focus primarily on screening predictors while treating responses as fixed, leaving the joint screening of both predictors and responses underexplored. As a result, even after predictor screening, computation remains intractable for multimodal data analysis. For instance, when jointly analyzing genome-wide DNA methylation (hundreds of thousands of CpG sites) and transcriptomic data (tens of thousands of genes), screening only predictors still leaves tens of thousands of responses, producing a prohibitively large cross-correlation matrix that is both computationally and memory intensive.


In practice, screening both sides is critical: it improves the detection of truly relevant variable pairs by enhancing the signal-to-noise ratio and uncovering latent regulatory structures. Partial correlation-based methods \citep{ke2022high}, building on earlier univariate approaches \citep{buhlmann2010variable}, have also been adapted to joint data. While these methods are theoretically appealing under Gaussian assumptions, their reliance on structural conditions such as partial faithfulness can be fragile in the presence of multicollinearity or interactions, often leading to missed associations.


To overcome these challenges, we propose a novel method using marginal correlations for joint data analysis: Graph Independence Dual Screening (GIDS). GIDS simultaneously screens predictors and responses while mitigating noise from spurious correlations, leveraging the graph structure of the marginal cross-correlations with memory parsimony. This graph-based approach allows for the identification of correlated modules or pathways, which often yield more interpretable and biologically meaningful results, particularly in multi-omics or imaging-omics studies. Importantly, we provide theoretical guarantees for GIDS by establishing the sure screening property and exact recovery with high probability.

The organization is as follows. In Section \ref{sec:2}, we introduce the background of screening for joint data sets, a measurement and algorithm for GIDS to detect subgraphs in which variables have strong associations. In addition, we evaluate GIDS using synthetic datasets and apply it to the motivating dataset in Section~\ref{sec:4}. Next, we introduce a set of assumptions and theoretical results based on the assumption in Section \ref{sec:5}. Then, the paper concludes with a discussion.

\section{Methods}
\label{sec:2}

\begin{figure}[h]
\centering
\includegraphics[width=1\textwidth]{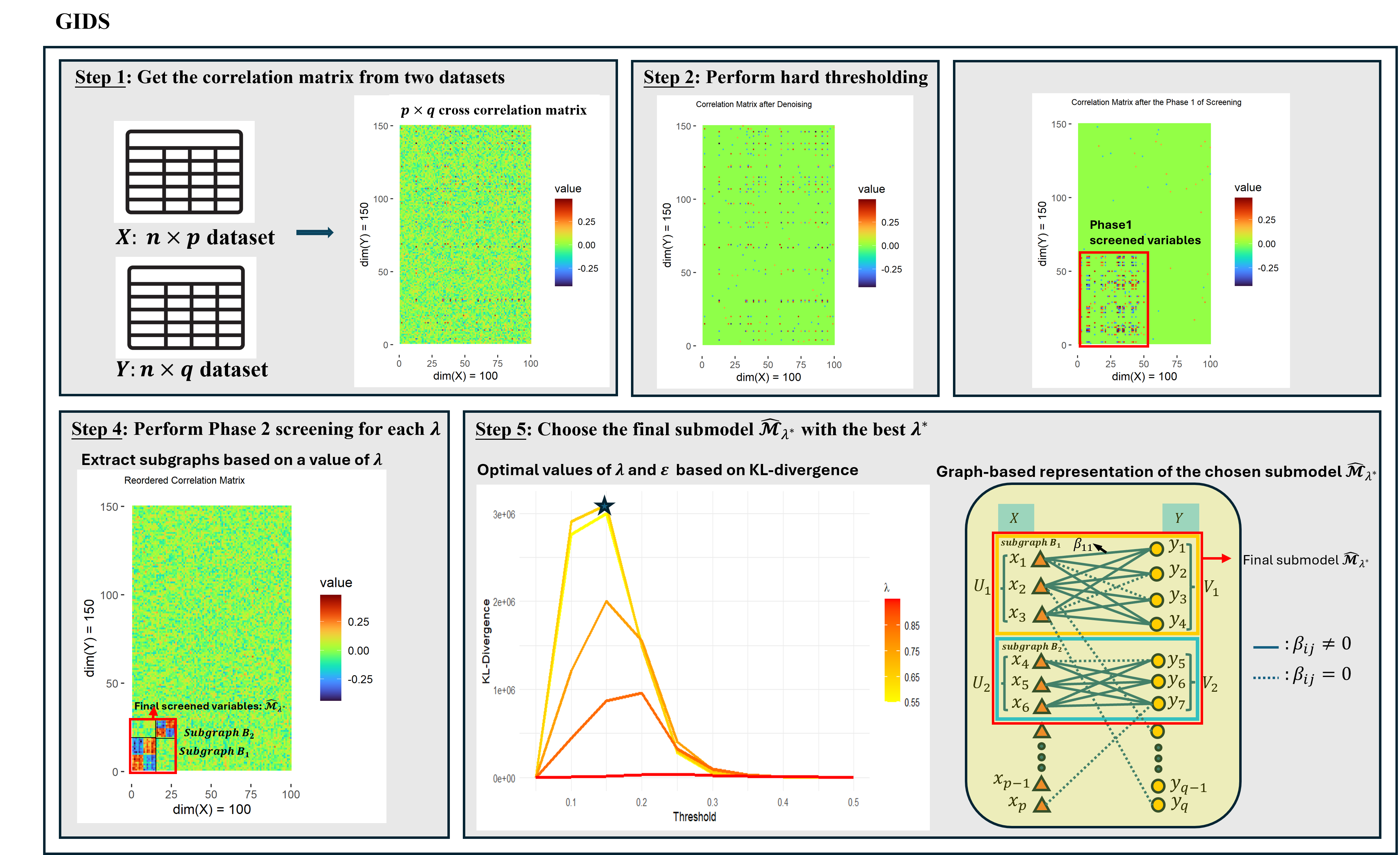}
\caption{Pipeline for screening true positive variables via GIDS. Step 1 presents the sample correlation matrix computed from two joint datasets, which is too large to be stored. Next, Step 2 illustrates the hard thresholded cross correlation matrix, where most spurious correlations are removed. Step 3 demonstrates Phase 1 screened variables. Step 4 displays the final screened variables (submodel $\widehat{\mathcal{M}}_{\lambda^\star}$) from Phase 2 screening. Finaly, Step 5 shows the selection of optimal tuning parameter $\lambda$ and the graph-based representation of the chosen submodel.}
\label{fig:1}
\end{figure}

To motivate our method, we begin by outlining the overall pipeline of the GIDS framework (Figure \ref{fig:1}). The pipeline is designed to handle ultra-high-dimensional multimodal data by reducing noise, compressing, and iteratively screening variables in a scalable manner. Starting with the sample cross-correlation matrix, which is often too large to store directly, the framework applies histogram-based compression and distributional modeling to separate signal from noise. A hard-thresholding step then removes spurious correlations, after which screening is conducted in two stages: a coarse Phase 1 filter that reduces the candidate space, followed by a fine-grained Phase 2 refinement across multiple tuning parameters. The process concludes by selecting the optimal tuning parameter and constructing graph-based submodels that reveal interpretable modules across the two data modalities. This systematic progression ensures both computational tractability and biological interpretability, supporting regression analysis and revealing interacting modules or pathways between the two variable sets.


\subsection{Bipartite Graph Structure for Joint Data Sets}
\label{subsec:21}
In high-dimensional association analysis, a common objective is to identify variable pairs $(i,j)$ such that $X_i$ and $Y_j$ are significantly correlated. A standard approach defines the true submodel as the set $\{(i,j) : \beta_{ij} \neq 0\}$, which flags all pairs with nonzero associations. However, this pairwise formulation often yields limited dimensionality reduction. This is because even a single significant entry in a pair $(i,j)$ leads to the inclusion of both variables $X_i$ and $Y_j$ in the associated variable sets, potentially inflating their size and diluting interpretability. This motivates the need for more parsimonious submodel formulations.

To mitigate this issue and better reflect underlying structured relationships, we adopt an alternative formulation based on structured subgraphs. Specifically, we make the use of a bipartite graph $B = (U, V; E)$, where $U$ and $V$ are disjoint sets of nodes corresponding to the variables of $\mathbf{X}$ and $\mathbf{Y}$, respectively, and $E$ is the set of binary edges representing the presence of non-zero correlations between the variables of $\mathbf{X}$ and $\mathbf{Y}$. We assume that there is a true submodel characterized by quasi-biclique subgraphs $\{B_c=(U_c, V_c; E_c)\}_{c=1}^C$, which reflect a non-square block-diagonal pattern. Here, $\{U_c\}_{c=1}^C$ and $\{V_c\}_{c=1}^C$ are disjoint subsets of $U$ and $V$, respectively, such that $\sum_{c=1}^C |U_c| \leq |U|$ and $\sum_{c=1}^C |V_c| \leq |V|$. We then define a true submodel as

$$
\mathcal{M}_\star = \{(i,j):e_{ij}\in E_c~\text{for~any~}c=1,\dots,C\},
$$
which serves as a parsimonious target capturing structured associations between subsets of variables.


In a joint dataset, important variable pairs tend to have high absolute sample correlation coefficients. Ideally, the distribution of the correlations for important variables should be distinguishable from that of unimportant ones. To formulate this, we model absolute sample correlations $|R_{ij}|$ using two distributions: $g_0$, corresponding to the null hypothesis, is concentrated near zero (e.g., a half-normal distribution), while $g_1$, representing the alternative hypothesis, is centered away from zero. We assume that if variables $X_i$ and $Y_j$ are associated, then $|R_{ij}| \sim g_1$; otherwise, $|R_{ij}| \sim g_0$. To model the edge presence probabilistically, we introduce a latent indicator variable $Z_{ij}$ that denotes whether $X_i$ and $Y_j$ are significantly associated. The magnitude of the sample correlation $|R_{ij}|$ is assumed to follow a mixture distribution:
\begin{eqnarray*}
|R_{ij}| \sim (1 - Z_{ij}) g_0 + Z_{ij} g_1,    
\end{eqnarray*}
Larger values of $|R_{ij}|$ indicate a higher probability that $Z_{ij} = 1$, signaling the likely presence of an edge in $E$, while smaller values suggest $Z_{ij} = 0$, indicating the absence of association.

Now, we connect the concept of latent indicator $Z_{ij}$ to a bipartite graph. This connection facilitates the modeling of $\mathcal{M}_\star$ by naturally aligning $Z_{ij}$ with the notion of an edge in a bipartite graph. For each pair $(i, j) \in U \otimes V$, $Z_{ij}$ is either 0 or 1 based on the existence of a binary edge $e_{ij}$ as follows:
\begin{eqnarray*}
Z_{ij}=1, & \text{if } e_{ij} \in E \\
Z_{ij}=0, & \text{otherwise}
\end{eqnarray*}
 where $I\otimes J$ for two sets $I$ and $J$ denotes the Cartesian product $\{(i,j):i\in I~\text{and}~j\in J\}$. The edge $e_{ij}$ is not observed as the indicator $Z_{ij}$ is not. Instead, we treat the absolute sample correlation $|R_{ij}|$ as a continuous proxy for the latent edge $e_{ij}$, using it to infer the likelihood of an edge's presence. The edge set $E_c$ for the $c$th subgraph can be natrually defined based on the latent indicators $Z_{ij}$ as $E_c = \{e_{ij}:Z_{ij}=1~\text{for}~(i,j)\in U_c\otimes V_c\}$. To reflect the expected structural sparsity, we assume that the probability of the presence of an edge is substantially higher within subgraphs than outside. Formally, for constants $0 < \delta_1, \delta_2 < 1$, we posit:
$$
\mathbb{P}(Z_{ij} = 1 \mid (i,j) \in U_c \otimes V_c~\text{for some } c) = 1 - \delta_1, \quad
\mathbb{P}(Z_{ij} = 1 \mid (i,j) \notin U_c \otimes V_c~\text{for any } c) = \delta_2,
$$

with $1 - \delta_1 \gg \delta_2$, representing that connections in a subgraph are denser than those out of subgraphs. In this framework, $\delta_1$ represents the probability of a false negative (i.e., an edge in a subgraph exhibiting null-like correlation), while $\delta_2$ reflects the probability of a false positive outside the subgraphs.

The goal of the screening method is to identify a set $\widehat{\mathcal{M}}$ such that $\mathcal{M}_\star \subset \widehat{\mathcal{M}}$, with smaller $\widehat{\mathcal{M}}$ being more desirable from a precision standpoint. Consistent with the screening literature, we also aim to recover the associated variable sets $I_X = \cup_{c=1}^C U_c$ and $I_Y = \cup_{c=1}^C V_c$, which represent the true positive subsets of ${\bf X}$ and ${\bf Y}$, respectively.

\subsection{Define (Estimate) $\mathcal{M}_\star$ by $\lambda$-Density}


Unlike traditional one-dimensional screening, where hard-thresholding can effectively reduce the cardinality of $X$,  estimating $\mathcal{M}_\star$ in a two-dimensinoal space requires new heuristics to reduce the cardinalities of $X$ and $Y$. To this end, we introduce the notion of $\lambda$-density of a bipartite graph for the estimation of quasi-bicliques, which include a maximal amount of signals (i.e., $\sum_{(i,j)\in \cup_{c=1}^C E_c} |R_{ij}|/\sum_{(i,j)\notin \cup_{c=1}^C E_c} |R_{ij}|$ ) within the space of $\mathcal{M}_\star$, while only maintaining drastically reduced numbers of variables in $X$ and $Y$ in the following analyses. We define the $\lambda$-density of a  subgraph $H=(U',V';E')$ of $B$ as follows:
\begin{eqnarray}
    \text{den}_{\lambda}(H)=\text{den}_{\lambda}(U',V')  =\frac{\sum_{i\in U', j\in V'}|R_{ij}^\varepsilon|}{(|U'||V'|)^\lambda}, \label{eq:den}
\end{eqnarray}
for a $\lambda$ in $[0.5, 1)$. The edge set can be surpressed in the notation such as $\text{den}_{\lambda}(H)=\text{den}_{\lambda}(U',V')$. To improve the performance of screening, we use truncated absolute sample correlations such that $|R_{ij}^\varepsilon|=|R_{ij}|I(|R_{ij}|>\varepsilon)$ for a cutoff $\varepsilon>0$. These are a denoised version of absolute sample correlation, which represent continuous proxies of latent edges. The $\lambda$-density is the sum of all truncated absolute sample correlations in a subgraph divided by its size, defined by $(|U'||V'|)^\lambda$. The $\lambda$-density can capture the strength of associations between two sets on each side of a bipartite graph. For example, $1$-density represents the sample mean of $|R_{ij}^{\varepsilon}|$ of a subgraph $H$. In addition, the $\lambda$-density can be interpreted as an approximation of the penalized canonical correlation value under the existence of subsets of significantly associated variables \citep{park2025graph}.

Maximizing $\lambda$-density can be written as minimizing the negative logarithm of the $\lambda$-density such that
$$
\underset{U'\subset U,~V'\subset V}{\text{min} }\left( - \log \text{den}_{\lambda}(U',V')\right) = \underset{U'\subset U,~V'\subset V}{\text{min}}\left(- \log \left(\sum_{i \in U', j \in V'} |R_{ij}^\varepsilon| \right) + \lambda \log |U'||V'|\right),
$$
which mirrors the structure of penalized objective functions commonly used in high-dimensional statistical estimation, particularly in sparse regression methods such as the lasso. The first term, $- \log \left(\sum |R_{ij}^\varepsilon| \right)$, encourages the inclusion of subgraphs with large cumulative association strength, which is analogous to maximizing likelihood or goodness-of-fit. The second term, $\lambda \log |U'||V'|$, acts as a complexity $\ell_0$-penalty of the sizes of the node sets, discouraging the selection of large subgraphs unless they demonstrate sufficiently strong internal coherence.

This formulation reflects the same core principle as lasso and ridge regression, which balance a data fidelity term (e.g., squared error loss) with a sparsity-inducing penalty (e.g., $\ell_1$ norm of the coefficients). In our case, the log-cardinality penalty $\log |U'||V'|$ discourages overly large bicliques by penalizing their size in a smooth, convex manner. This yields a trade-off between the total strength of signal and the size of the subgraph, promoting the discovery of compact, highly connected subgraphs based on a parsimonious model.

From this viewpoint, the minimization of $-\log \text{den}_\lambda(H)$ over all candidate subgraphs can be interpreted as a graph-based analogue of sparse regression: rather than selecting individual variables via coefficient shrinkage, we are selecting dense subgraphs via a penalty on block size. The parameter $\lambda$ thus serves a role similar to the regularization parameter in penalized models, controlling the balance between detection sensitivity and structural parsimony. However, maximizing $\text{den}_\lambda(H)$ remains an NP-hard problem, as it involves identifying the densest subgraph, which is a well-known computationally intractable task. To address this challenge, we propose a greedy algorithm that constructs a sequence of high-density subgraphs across varying subgraph sizes, providing a scalable approximation to the underlying optimization problem.

The following proposition shows that $\lambda$-density provides a principled and well-defined measure for detecting a true subgraph when there exists a block of consistently strong associations. It formalizes the intuition that, under homogeneity of signal strength, the true subgraph can be recovered by maximizing $\lambda$-density across candidate subgraphs.
\begin{proposition}
Let $B=(U,V;E)$ be a bipartite graph and let
\[
w_{ij}:=|\rho_{ij}|,
\qquad (i,j)\in U\times V,
\]
be the absolute correlation coefficients. For nonempty $S\subseteq U$ and $T\subseteq V$, write
\[
W(S,T):=\sum_{i\in S,\,j\in T}w_{ij},
\qquad
m(S,T):=|S||T|,
\qquad
\operatorname{den}_{\lambda}(S,T):=\frac{W(S,T)}{m(S,T)^{\lambda}},
\qquad \lambda\in[1/2,1).
\]
Fix a candidate quasi-biclique $B_0=(U_0,V_0)$ with $W_0:=W(U_0,V_0)>0$ and $m_0:=|U_0||V_0|$. For every competing nonempty rectangle $A=(S,T)\neq B_0$, set $W_A:=W(S,T)$ and $m_A:=m(S,T)$. Define
\[
\Lambda(B_0):=\left\{\lambda\in[1/2,1):
\frac{W_0}{m_0^{\lambda}}\geq \frac{W_A}{m_A^{\lambda}}
\text{ for all nonempty }A=(S,T)\neq B_0\right\}.
\]
Then $B_0$ maximizes the $\lambda$-density for every $\lambda\in\Lambda(B_0)$ if and only if $\Lambda(B_0)$ is nonempty.
\label{prop:2}
\end{proposition}

\subsection{Greedy Algorithm for Graph Independence Dual Screening}

To improve the performance of the greedy selection strategy, we objectively choose a thresholding level $\varepsilon$ aimed at better distinguishing true positive nodes from those not in the target sets. Thresholding helps suppress the influence of weak or spurious associations, thereby enhancing the signal-to-noise ratio. The optimal value of $\varepsilon$ can differ based on tasks with different signal-to-noise ratios, which is the ratio of important and unimportant variables in a phase of a screening procedure. We use two different thresholds for two steps: i) hard thresholding $|R_{ij}|$ such that $|R_{ij}|I(|R_{ij}|>\varepsilon_1)$ with level $\varepsilon_1$ for reordering the rows and columns of the absolute sample correlation matrix, ii) cutoff $\varepsilon_2$ for binary edges $I(|R_{ij}|>\varepsilon_2)$ to test the significance of extracted subgraph. Hard thresholding in step (i) reduces the variance of the absolute sample correlations for noise variables from $O(n^{-1})$ to $O(\exp\left(-\Omega(n)\right))$, which facilitates screening and lowers the error rate for achieving the sure screening property, from $O\!\left(\exp(-\Omega(\sqrt{n}))\right)$ to $O\!\left(\exp(-\Omega(n))\right)$. 

First, we calculate the absolute sample correlation coefficients and discretize them as histogram-valued data for faster and memory-efficient estimation. The histogram-valued data facilitate the calculation of likelihood function. With the simplified likelihood, we calculate the components of mixture distributions of noise and signal. Then, based on the estimates of mixture components, we calculate the denosing level $\varepsilon_1$ for the subgraph extraction and threshold $\varepsilon_2$ for binary edges to test the significance of extracted subgraph. Further details are provided in Supplementary Materials~\ref{sub:thresholding}.


The greedy algorithm performs a row and column exclusion process, yielding the rank of variables for both predictors and reponses. The greedy algorithm starts with node sets $U'\subset U$ and $V'\subset V$ and reduce the node sets to the target size $(p_{\text{target}}, q_{\text{target}})$ using aggregated row and column sums of the truncated absolute  correlation matrix $[|R^{\varepsilon_1|}_{ij}]_{i\in U',j\in V'}$, where $U' =\{u_1,u_2,\dots,u_{|U'|}\}$ and $V' =\{v_1,v_2,\dots,v_{|V'|}\}$. Prior to applying the algorithm, we specify the granularity of the algorithm $(k_1,k_2)$, which represent the numbers of row and column exclusion at each iteration, respectively. We then compute the row and column aggregate scores of the truncated absolute correlation matrix as
\begin{eqnarray*}
\mathbf{r.sum}_1 &=& \left(\sum_{j\in V'}|R_{u_1'j}^{\varepsilon_1}|, \sum_{j\in V'}|R_{u_2'j}^{\varepsilon_1}|, \dots, \sum_{j\in V'}|R_{u_{p_{initial}}'j}^{\varepsilon_1}|\right),\\
\mathbf{c.sum}_1 &=& \left(\sum_{i\in U'}|R_{iv_1'}^{\varepsilon_1}|, \sum_{i\in U'}|R_{iv_2'}^{\varepsilon_1}|, \dots, \sum_{i\in V}|R_{iv_{q_{initial}}'}^{\varepsilon_1}|\right).
\end{eqnarray*}
 The greedy algorithm initiates with the full node sets $(\widetilde{U}_1, \widetilde{V}_1)=(U',V')$ as active node sets at time $t = 1$. At each iteration $t$, the algorithm selects subsets of $k_1$ rows and $k_2$ columns with the lowest scores from the current active node sets, defined as
\begin{eqnarray*}
\bm{u}_t = \operatorname{argmin}^{(k_1)}_{u\in \widetilde{U}_t}(\mathbf{r.sum}_{t,u}), \quad
\bm{v}_t = \operatorname{argmin}^{(k_2)}_{v\in \widetilde{V}_t}(\mathbf{c.sum}_{t,v}),  
\end{eqnarray*}
where $\operatorname{argmin}^{(k)}_{i \in I} v_i$ denotes the indices of the $k$ smallest values of the vector $\bm{v}=(v_1,\dots,v_l)$ over the integer set $I\subset [l]$ and $\mathbf{r.sum}_{t,u}$ is the $u$th component of $\mathbf{r.sum}_t$. Then, excludes the rows $\bm{u}_t$ from $\widetilde{U}_t$ and update 
\begin{eqnarray*}
&&\widetilde{U}_{t+1} \leftarrow \widetilde{U}_t\backslash \bm{u}_t, \widetilde{V}_{t+1} \leftarrow \widetilde{V}_t, \mathbf{r.sum}_{t,\bm{u}_t}=0, \mathbf{r.sum}_{t+1} \leftarrow \mathbf{r.sum}_t,\\
&&\mathbf{c.sum}_{t,j} \leftarrow \mathbf{c.sum}_{t,j}-\sum_{u\in \bm{u}_t}|R_{uj}^{\varepsilon_1}|~\text{for~all}~j\in \widetilde{V}_{t+1}, \mathbf{r.sum}_{t+1} \leftarrow \mathbf{r.sum}_{t},
\end{eqnarray*}
if $\{\sum_{u\in \bm{u}_t}\mathbf{r.sum}_{t,u}/|\widetilde{V}_t||k_1| < \sum_{v \in \bm{v}_t}\mathbf{c.sum}_{t,v}/|\widetilde{U}_t||k_2|~\text{and}~|\widetilde{U}_t|>p_{\text{target}}\}~\text{or}~|\widetilde{V}_t|\leq q_{\text{target}}$. Otherwise, excludes the columns $\bm{v}_t$ from $\widetilde{V}_t$ and update 
\begin{eqnarray*}
&&\widetilde{U}_{t+1}\leftarrow \widetilde{U}_t,~\widetilde{V}_{t+1}\leftarrow \widetilde{V}_t\backslash \bm{v}_t,~\mathbf{c.sum}_{t,\bm{v}_t}=0,~\mathbf{c.sum}_{t+1} \leftarrow \mathbf{c.sum}_t\\
&&\mathbf{r.sum}_{t,i} \leftarrow \mathbf{r.sum}_{t,i}-\sum_{v\in \bm{v}_t}|R_{i v}^{\varepsilon_1}|~\text{for all}~i\in \widetilde{U}_{t+1},~\mathbf{r.sum}_{t+1} \leftarrow \mathbf{r.sum}_{t}.
\end{eqnarray*}
Then, $\mathbf{r.sum}_{t,\widetilde{U}_t}$ and $\mathbf{c.sum}_{t,\widetilde{V}_t}$ at time $t$ represent the row and column aggregate scores of the (truncated) absolute correlation matrix with the variables of ${\bf X}$ and ${\bf Y}$ indexed by $\widetilde{U}_t$ and $\widetilde{V}_t$, respectively. In addition, $\mathbf{r.sum}_t$ and $\mathbf{c.sum}_t$ have 0 on $U'\backslash\widetilde{U}_t$ and $V'\backslash\widetilde{V}_t$, respectively. The first phase ends at the earliest time $t=t'$ such that $|\widetilde{U}_{t'}|\leq p_{\text{target}}$ and $|\widetilde{V}_{t'}|\leq q_{\text{target}}$. We denote this process by $\text{Greedy}( (U',V'),(p_{target},q_{target}),(k_1,k_2))$, which returns the screened variables $\{(\widetilde{U}_t,\widetilde{V}_t)\}_{t=1}^{t'}$ and $\{\text{den}_0(\widetilde{U}_t,\widetilde{V}_t)\}_{t=1}^{t'}$.

As previously noted, storing the full correlation matrix for a joint dataset can be prohibitively memory-intensive. To address this, we implement two versions of the algorithm. The memory-efficient version avoids storing the entire absolute correlation matrix at the outset. Instead, during the updates of $\mathbf{r.sum}$ and $\mathbf{c.sum}$, the quantities $\sum_{v \in \mathbf{v}_t} |R_{iv}^{\varepsilon_1}|$ and $\sum_{u \in \mathbf{u}_t} |R_{uj}^{\varepsilon_1}|$ are recomputed directly from the dataset each time. The speed-optimized version, in contrast, prioritizes computational efficiency over memory usage. Here, the algorithm loads and reuses the pre-stored quantities for updating $\mathbf{r.sum}$ and $\mathbf{c.sum}$. This approach is feasible when the initial node sets $U'$ and $V'$ are small enough that the cross-correlation matrix can fit comfortably within available memory.

In the first phase (Phase 1), we reduce the variable dimensions from $(p, q)$ to $(p_{\text{phase1}}, q_{\text{phase1}})$ with the memory-efficent version of greedy algorithm, as the dimensions $(p, q)$ is too large to store the cross correlation matrix. The choice of $(p_{\text{phase1}}, q_{\text{phase1}})$ should reflect two considerations: (i) the final screened set cannot exceed this size, and (ii) the correlation matrix of this size must be storable within the workstation’s memory constraints. We perform the exclusion process with coarse granularity, removing $k_1\geq 1$ rows or $k_2\geq 1$ columns at each step. 

In the second phase (Phase 2), we perform subgraph extraction for further screening with a finer granularity ($k_1=k_2=1$). At this point, we utilize the stored absolute correlation matrix, enabling more fast computation. The algorithm in Phase continues the exclusion process for subgraph extractions pararelly over the values of $\bm{\lambda}=(\lambda_1,\dots,\lambda_g)$, considering the screened variables in Phase 1 as the initial active node set $(\widetilde{U}_{c,1}^
\lambda,\widetilde{V}_{c,1}^
\lambda)=(\widetilde{U}_{t'},\widetilde{V}_{t'})$ for $c=1$ and all $\lambda \in \bm{\lambda}$ at a finer granularity, setting $k_1=1$ and $k_2=1$. Then, based on the set of sequence $\{(\widetilde{U}_{c,t},\widetilde{V}_{c,t})\}$, extract the subgraph of node set $(\widetilde{U}_{c,t_{\lambda,c}}^
\lambda,\widetilde{V}_{c,t_{\lambda,c}}^
\lambda)$ such that
\begin{eqnarray}
    t_{\lambda,c}=\text{argmax}_{t=1,2,\dots,|\widetilde{U}_{c,1}^
\lambda|+|\widetilde{V}_{c,1}^
\lambda|-2} \text{den}_\lambda (\widetilde{U}_{c,t}^
\lambda,\widetilde{V}_{c,t}^
\lambda).\nonumber
\end{eqnarray}
Next, to extract the $(c+1)$st subgraph, we repeat the procedure described above with an initial active node sets $\widetilde{U}_{c+1,1}^\lambda = \widetilde{U}_{c,1}^\lambda \backslash \widetilde{U}_{c,t_{\lambda,c}^\star}^\lambda$ and $\widetilde{V}_{c+1,1}^\lambda = \widetilde{V}_{c,1}^\lambda \backslash \widetilde{V}_{c,t_{\lambda,c}^\star}^\lambda$, until a stop criterion is satisfied. Once the stopping condition is first satisfied during the $(c+1)$st subgraph extraction, we set $C = C_\lambda$ to indicate the total number of extracted subgraphs for the given $\lambda$.

We now describe the stopping criterion for the extraction procedure by testing the significance of the $c$th extracted subgraph.  Under the null hypothesis (i.e. there are no significant associations between $X$ and $Y$), the full graph with binary edges $\{I(|R_{ij}|>\varepsilon)\}_{i,j}$ can be modeled as an Erdős–Rényi model regardless of the value of $\varepsilon$. However, selecting an appropriate threshold remains critical, as it governs the trade-off between sensitivity and precision in identifying signal edges. To optimally balance these, the threshold $\varepsilon_2$ is selected by maximizing the F1 score. Specifically, we accept an extracted subgraph $B_c = (U_c,V_c;E_c)$ if
$$h(U_c,V_c):=\exp \left( - |U_c| |V_c| \left[ D(\gamma_c^{\varepsilon_2} \| \gamma_0^{\varepsilon_2}) - \frac{\log \left( |U|e/|U_c| \right)}{|V_c|} - \frac{\log \left( |V|e/|V_c| \right)}{|U_c|} \right] \right) < \delta,$$
where $\gamma_c^{\varepsilon_2} = \sum_{i\in U_c,j\in V_c}I(|R_{ij}|>\varepsilon_2)/|U_c||V_c|$, $\gamma_0^{\varepsilon_2} = \sum_{i\in U,j\in V}I(|R_{ij}|>\varepsilon_2)/|U||V|$, $D(a\|b) = a\log (a/b) + (1-a)\log ((1-a)/(1-b))$. The theoretical property and justification of this stop criterion will be introduced in the next section. 

Finally, we introduce the precedure to choose the best $\lambda$. The value of chosen $\lambda$ for $\varepsilon$ based on $\{I(|R_{ij}|>\varepsilon)\}_{i,j}$ can vary depending the value of $\varepsilon$. For example, for a big subgraph, which consists of two highly related subgraphs, a small $\lambda$ can be chosen for a small $\varepsilon$, which corresponds to a big subgraph, while two splited subgraphs of it can be chosen with higher $\varepsilon$. The optimal $\lambda$ maximizes the KL-divergence such that
\begin{eqnarray}
    \lambda^\star =  \text{argmax}_{\lambda\in \bm{\lambda}} 
    \text{max}_{\varepsilon \in \bm{\varepsilon}}D_{\mathrm{KL}}\left(P_{\lambda,\varepsilon} \| Q_\varepsilon \right),\label{eq:normKL}
\end{eqnarray}
where $D_{\mathrm{KL}}$ is the KL-divergence and $P_{\lambda,\varepsilon}$ and $Q_\varepsilon$ are the target and reference distributions. The details of these are described in Supplementary Materials \ref{sub:lambda}. The extracted subgraphs are $\widehat{B} = \widehat{B}^{\lambda^\star} =\{\widehat{B}_c^{\lambda^\star}=(\widehat{U}_c^{\lambda^\star},\widehat{V}_c^{\lambda^\star};\widehat{E}_c^{\lambda^\star})\}_{c=1}^{C_{\lambda^\star}}$, where the edge sets for the final model are calibrated to be $\widehat{E}_c^{\lambda^\star} = \{e_{ij}:|R_{ij}| > \varepsilon^\star~\text{for}~(i,j)\in \widehat{U}_c^{\lambda^\star}\otimes\widehat{V}_c^{\lambda^\star}\}$ for $c=1,\dots,C_{\lambda^\star}$. Lastly, the algorithm returns the final submodel and the indice of screened variables such that $\widehat{\mathcal{M}} = \{(i,j):e_{ij}\in \widehat{E}_c^{\lambda^\star}~\text{for~any}~c=1,\dots,C_{\lambda^\star}\}$, $\widehat{I}_X = \cup_{c=1}^{C_{\lambda^\star}}\widehat{U}_c^{\lambda^\star}$, and $\widehat{I}_Y = \cup_{c=1}^{C_{\lambda^\star}}\widehat{V}_c^{\lambda^\star}$.

\begin{algorithm}[h]
\caption{: Greedy$((U',V'),(p_{\text{target}},q_{\text{target}}),(k_1,k_2),\varepsilon_1,\text{isMP})$: row–column exclusion}
\begin{algorithmic}[1]
        \State \textbf{Input:} joint data set $({\bf X, Y})$, initial node sets $(U',V')$, $(p_{\text{target}},q_{\text{target}})$, $(k_1,k_2)$, $\varepsilon_1$, \text{isMP}
		\State \textbf{Output:} Screened variables and the order of subsets for excluded variables
        \State If \text{isMP} = False, store $|R_{ij}^{\varepsilon_1}|$ for $i\in U'$ and $j\in V'$ in the memory
        \State  Set $(\widetilde{U}_1,\widetilde{V}_1)=(U',V')$
        \State Calcultate row and column sums
$$\mathbf{r.sum}_1 = \left(\sum_{j\in V'}|R_{u_1'j}^{\varepsilon_1}|, \dots, \sum_{j\in V'}|R_{u_{|U'|}'j}^{\varepsilon_1}|\right),
\mathbf{c.sum}_1 = \left(\sum_{i\in U'}|R_{iv_1'}^{\varepsilon_1}|, \dots, \sum_{i\in V}|R_{iv_{|V'|}'}^{\varepsilon_1}|\right)$$
\State Set $t' = \lceil(|U'|-p_{\text{target}})/k_1\rceil+\lceil(|V'|-q_{\text{target}})/k_2\rceil$
\For{$t = 1,2,\dots ,t'$}
\State Calculate the smallest $k_1$ rows and $k_2$ columns $$\bm{u}_t = \operatorname{argmin}^{(k_1)}_{u\in \widetilde{U}_t}(\mathbf{r.sum}_{t,u}), \quad
\bm{v}_t = \operatorname{argmin}^{(k_2)}_{v\in \widetilde{V}_t}(\mathbf{c.sum}_{t,v})$$
\State Update active sets and row and columns sums:
\State If $\{\sum_{u\in \bm{u}_t}\mathbf{r.sum}_{t,u}/|\widetilde{V}_t||k_1| < \sum_{v \in \bm{v}_t}\mathbf{c.sum}_{t,v}/|\widetilde{U}_t||k_2|~\text{and}~|\widetilde{U}_t|>p_{\text{target}}\}~\text{or}~|\widetilde{V}_t|\leq q_{\text{target}}$
\begin{eqnarray*}
&&\widetilde{U}_{t+1} \leftarrow \widetilde{U}_t\backslash \bm{u}_t, \widetilde{V}_{t+1} \leftarrow \widetilde{V}_t, \mathbf{r.sum}_{t,\bm{u}_t}=0, \mathbf{r.sum}_{t+1} \leftarrow \mathbf{r.sum}_t,\\
&&\mathbf{c.sum}_{t,j} \leftarrow \mathbf{c.sum}_{t,j}-\sum_{u\in \bm{u}_t}|R_{uj}^{\varepsilon_1}|~\text{for~all}~j\in \widetilde{V}_{t+1}, \mathbf{r.sum}_{t+1} \leftarrow \mathbf{r.sum}_{t},
\end{eqnarray*}
\State otherwise,
\begin{eqnarray*}
&&\widetilde{U}_{t+1}\leftarrow \widetilde{U}_t,~\widetilde{V}_{t+1}\leftarrow \widetilde{V}_t\backslash \bm{v}_t,~\mathbf{c.sum}_{t,\bm{v}_t}=0,~\mathbf{c.sum}_{t+1} \leftarrow \mathbf{c.sum}_t\\
&&\mathbf{r.sum}_{t,i} \leftarrow \mathbf{r.sum}_{t,i}-\sum_{v\in \bm{v}_t}|R_{i v}^{\varepsilon_1}|~\text{for all}~i\in \widetilde{U}_{t+1},~\mathbf{r.sum}_{t+1} \leftarrow \mathbf{r.sum}_{t}.
\end{eqnarray*}
\EndFor
         \State Return $\{(\widetilde{U}_t,\widetilde{V}_t)\}_{t=1}^{t'}$ and $\{\text{den}_0(\widetilde{U}_t,\widetilde{V}_t)\}_{t=1}^{t'}$
	\end{algorithmic}
\label{algo1}
\end{algorithm}

\begin{algorithm}[h]
\caption{: GIDS procedure}
\begin{algorithmic}[1]
        \State \textbf{Input:} joint data set $({\bf X},{\bf Y})$, $\bm{\lambda} = \{\lambda_1,\dots,\lambda_{g_1}\}$, $(p_{\text{phase1}},q_{\text{phase1}})$, $(k_1,k_2)$, $\delta$
		\State \textbf{Output:} A set of subgraphs and screened variables
        \State Calculate absolute sample correlations and summarize them as histogram-valued data
        \State Calculate the thresholding level $\varepsilon_1$ for the greedy algoritm and $\varepsilon_2$ for binary edge 
         \State $\{(\widetilde{U}_t,\widetilde{V}_t)\}_{t=1}^{t'}$ $\leftarrow$ Greedy($(U,V)$,~$(p_{\text{phase1}},q_{\text{phase1}})$,~$(k_1,k_2)$,~$\varepsilon_1$, \text{isMP}=True) \Comment{Phase 1} 
        \For{$g = 1,2,\dots ,g_1$} \Comment{Phase 2}
        \State Set $c=1$ and $\widetilde{U}^{\lambda_g}_{c,1}=\widetilde{U}_{t'}$ and $\widetilde{V}^{\lambda_g}_{c,1}=\widetilde{V}_{t'}$
        \While{}
            \State $\{(\widetilde{U}^{\lambda_g}_{c,t},\widetilde{V}^{\lambda_g}_{c,t})\}_{t=1}^{|\widetilde{U}^{\lambda_g}_{c,1}|+|\widetilde{V}^{\lambda_g}_{c,1}|-2}$ $\leftarrow$ Greedy($(\widetilde{U}^{\lambda_g}_{c,1},\widetilde{V}^{\lambda_g}_{c,1})$,~$(1,1)$,~$(1,1),~\varepsilon_1$, \text{isMP}=False)
            \State $    t_{\lambda_g,c}=\text{argmax}_{t=1,2,\dots,|\widetilde{U}_{c,1}^{\lambda_g}|+|\widetilde{V}_{c,1}^{\lambda_g}|-2} \text{den}_\lambda(\widetilde{U}_{c,t}^{\lambda_g},\widetilde{V}_{c,t}^{\lambda_g})$
            \State $(\widetilde{U}_{c+1,1}^{\lambda_g},\widetilde{V}_{c+1,1}^{\lambda_g})\leftarrow (\widetilde{U}_{c,1}^{\lambda_g} \backslash \widetilde{U}_{c,t_{\lambda_g,c}}^{\lambda_g},\widetilde{V}_{c,1}^{\lambda_g} \backslash \widetilde{V}_{c,t_{\lambda_g,c}}^{\lambda_g})$ and $c\leftarrow c+1$ 
        \EndWhile~{$h(\widetilde{U}_{c-1,t_{\lambda_g,c-1}}^{\lambda_g},\widetilde{V}_{c-1,t_{\lambda_g,c-1}}^{\lambda_g})<\delta$}        
        \EndFor
        \State Find the best $\lambda^{\star}$ based on \eqref{eq:normKL}
        \State Return Phase 1 screened variables $(\widetilde{U}_{t'},\widetilde{V}_{t'})$ and the extracted subgraphs $\{(\widetilde{U}_{c,t_{\lambda^\star,c}}^{\lambda^\star},\widetilde{V}_{c,t_{\lambda^\star,c}}^{\lambda^\star})\}_{c=1}^C$
	\end{algorithmic}
\label{algo1}
\end{algorithm}

\section{Simulation Studies}
\label{sec:4}

We numerically evaluate the performance of the proposed GIDS approach and benchmark it with Distance Correlation Sure Independence Screening (DC-SIS) \citep{li2012feature}, Ball Correlation Sure Independence Screening (Bcor-SIS) \citep{pan2019generic}, and Projection Correlation Screening (PC-Screen) \citep{liu2022model}. We simulate $n$ samples ($n=200$) of $X$ and $Y$ based on
\begin{eqnarray*}
\begin{pmatrix}
    X     \\
    Y     \\
\end{pmatrix}\sim \mathcal{N}(\mathbf{0}_{p+q},\Sigma),~\Sigma=\begin{pmatrix}
    \Sigma_X & \Sigma_{XY}     \\
     \Sigma_{XY}^\top &  \Sigma_Y     \\
\end{pmatrix}
\end{eqnarray*}
where $\text{diag}(\Sigma) = \mathbf{1}_{p+q}$ and a bipartite graph for $\Sigma_{XY}$ has latent biclique subgraphs $B_c = (U_c,V_c;E_c)$ for $c=1,2$. $X_i$ and $Y_j$ have a nonzero correlation $\rho_{ij} \neq 0$, if for $i\in U_1$ and $j\in V_1$; otherwise, $\rho_{ij}=0$. We generate samples directly from the joint distribution, rather than using regression equations, because of their one-to-one correspondence for Guassian variables and the more intuitive interpretation of the joint Gaussian model. As the benchmarking methods are supposed to screen only predictors, we compare the performances of methods with a focus on the predictor screening.


We conduct experiments under four different scenarios, formed by the cross product of four settings for variable dimensions and two settings for the true variable size and sample size. To examine the effect of response variable dimensions on predictor screening, we consider four setups: A) $(p,q)=(1000,1500)$; B) $(p,q)=(4000,1500)$; C) $(p,q)=(1000,5000)$; D) $(p,q)=(4000,5000)$ with the fixed sizes of two subgraphs such that $(|U_1|,|V_1|)=(20,30)$ and $(|U_2|,|V_2|)=(30,40)$ with $n=200$. The correlations are specified as $\rho_{ij}=0$ if $Z_{ij}=0$, $\rho_{ij}=0.3$ if $Z_{ij}=1$, with $\delta_1=0.04$ and $\delta_2=0$ in all cases. For each setup, we select the optimal tuning parameter from the range 0.5 to 0.9 (incremented by 0.1), based on the KL divergence.

For each setting, we generate 100 data sets to evaluate the performance of GIDS and the other three methods by assessing the sensitivity, precision, and F1-score of GIDS and other methods for the two settings. For GIDS, we consider both predictors and responses for mesaurement, whereas the competing methods only account for predictors. In the context of classification, \textit{sensitivity}  measures the proportion of correlated $X$ and $Y$ pairs that are correctly identified by the model. It is defined as:
\begin{eqnarray*}
\text{Sensitivity} =  \frac{ |\widehat{I}_X\cap I_X|+|\widehat{I}_Y\cap I_Y|}{\underbrace{|\widehat{I}_X\cap I_X|+|\widehat{I}_Y\cap I_Y|}_{\text{number of true positives}} + \underbrace{|I_X\backslash \widehat{I}_X|+|I_Y\backslash \widehat{I}_Y|}_{\text{number of false negatives}} }.    
\end{eqnarray*}
On the other hand, \textit{precision}  measures the proportion of true uncorrelated $X$ and $Y$ pairs that are correctly identified by the model. It is defined as:
\begin{eqnarray*}
\text{Precision} = \frac{ |\widehat{I}_X\cap I_X|+|\widehat{I}_Y\cap I_Y|}{\underbrace{|\widehat{I}_X \cap I_X|+|\widehat{I}_Y \cap I_Y|}_{\text{number of true positives}} + \underbrace{|\widehat{I}_X\backslash I_X|+|\widehat{I}_Y\backslash I_Y|}_{\text{number of false positives}}}.   
\end{eqnarray*}
The \textit{F1-score} is the harmonic mean of sensitivity and precision. It provides a single metric that balances the trade-off between these two measures, especially in cases where one is significantly higher than the other. A high F1-score indicates that the method achieves both high sensitivity (few false negatives) and high precision (few false positives), making it a reliable measure of overall performance. It is particularly useful in imbalanced settings where accuracy alone may be misleading. It is written as
\begin{eqnarray*}
\text{F1 score}=\frac{2\cdot\text{Precision}\cdot \text{Sensitivity}}{\text{Sensitivity} + \text{Precision} }.
\end{eqnarray*}

We evaluate the screening performance of four methods for two screening sizes in predictor screening, with $d_1 = 2\lceil n/\log n\rceil$ and $d_2 = 4\lceil n/\log n\rceil$, where $\lceil n/\log n\rceil=38$ for $n=200$. Table \ref{table:1} shows that GIDS, particularly in Phase 1, consistently achieves the highest sensitivity, precision, and F1-scores across cases and screening sizes. In Cases (A) and (B), where the dimension of the responses is relatively low, the three baseline methods (DC-SIS, Bcor-SIS, PC-Screen) perform competitively with each other, and although GIDS still achieves the best results except for the sensitivity for Cases (A), the margin over the baselines is modest. In contrast, in Cases (C) and (D), which involve relatively higher dimensional responses, the performance gap becomes much larger—GIDS substantially outperforms the baselines across all metrics, highlighting its robustness in more challenging, high-dimensional responses settings.

Overall, the results confirm that GIDS delivers near-optimal performance and consistently surpasses existing methods. In addition, it shows that the performance of the three benchmarking methods is affected by the dimension of responses rather than the dimension of predictors, due to the lack of a thresholding mechanism for high-dimensional responses.

\begin{table}[]
\caption{Percentages of sensitivity, precision, and F1 Score of GIDS and the other three methods for four cases based on the different numbers of total variables, true variables, and samples with two screening sizes $d_1=76$ and $d_2=152$. }
\centering
\begin{tabular}{clcccccccccc}\hline
         &             Measurement  & \multicolumn{2}{c}{DC-SIS} & & \multicolumn{2}{c}{Bcor-SIS}  & & \multicolumn{2}{c}{PC-Screen}   && GIDS      \\
\cline{3-4} 
\cline{6-7} 
\cline{9-10} 
         &                  & $d_1$ & $d_2$ && $d_1$ & $d_2$  && $d_1$ & $d_2$   &&       \\\hline\hline
        Case (A) &  Sensitivity & 0.959 & 0.987 && 0.823 & 0.943 && 0.959 & $\mathbf{0.987}$ && 0.973 \\
        $p=1000$ & Precision   & 0.631 & 0.324 && 0.542 & 0.310 && 0.631 & 0.325 && $\mathbf{1.000}$\\
        $q=1500$ & F1 Score    & 0.761 & 0.489  && 0.653 & 0.467 && 0.761 & 0.489 && $\mathbf{0.986}$\\\hline\hline
        Case (B) & Sensitivity & 0.910 & 0.954 && 0.638 & 0.790 && 0.920 & 0.960 && $\mathbf{0.962}$\\
        $p=4000$ & Precision   & 0.599 & 0.314 && 0.420 & 0.260 && 0.605 & 0.316 && $\mathbf{1.000}$ \\
        $q=1500$ & F1 Score    & 0.722 & 0.472 && 0.507 & 0.391 && 0.730 & 0.475 && $\mathbf{0.980}$\\\hline\hline
        Case (C) & Sensitivity & 0.466 & 0.625 && 0.407 & 0.632 && 0.760 & 0.875 && $\mathbf{0.959}$ \\
        $p=1000$ & Precision   & 0.306 & 0.206 && 0.268 & 0.208 && 0.500 & 0.288 && $\mathbf{1.000}$ \\
        $q=5000$ & F1 Score    & 0.370 & 0.309 && 0.323 & 0.313 && 0.603 & 0.433 && $\mathbf{0.979}$ \\\hline\hline
        Case (D) & Sensitivity & 0.279 & 0.390 && 0.212 & 0.345 && 0.626 & 0.731 && $\mathbf{0.962}$ \\
        $p=4000$ & Precision   & 0.184 & 0.128 && 0.140 & 0.113 && 0.412 & 0.240 && $\mathbf{1.000}$ \\
        $q=5000$ & F1 Score    & 0.222 & 0.193 && 0.168 & 0.171 && 0.497 & 0.362 && $\mathbf{0.980}$ \\\hline
\end{tabular}

\label{table:1}
\end{table}

\section{Application to DNA-methylation and transcriptomics interaction analysis in the ADNI data set}
\label{sec:4}

In this section, we analyze joint data sets of Alzheimer's Disease Neuroimaging Initiative (ADNI), which provides rich multiomics datasets that enable integrative analysis of molecular mechanisms underlying Alzheimer's disease. In particular, genome-wide DNA methylation data measured using high-density Illumina EPIC arrays offer epigenetic profiles across over 865353 CpG sites, while transcriptomic data acquired through RNA sequencing capture gene expression levels for up to 49386 transcripts. Note that the covariance of these two data sets cannot be stored in most working stations. These two high-dimensional datasets, derived from partially overlapping cohorts with 312 individuals, support joint analysis of epigenetic regulation and transcriptional activity in the progression of Alzheimer's disease. By leveraging the complementary nature of these data types, researchers can identify coordinated patterns such as methylation-driven gene silencing or activation and improve biomarker discovery for early diagnosis and disease subtyping.

In this study, our objective is to systematically screen associated variables between DNA methylation and gene expression. For the Phase 1 screening of GIDS, we set the dimensions to $(p_{\text{phase1}},~q_{\text{phase1}})=(10000,~6000)$, corresponding to DNA methylation and gene expression, respectively. Prior to applying the greedy algorithm, we estimate the noise and signal distributions for the entire correlation matrix. The proportion of the signal distribution is estimated as $\widehat{\pi}_1 = 0.0006$, with mean $\widehat{\mu}_1 = 0.165$. Thresholds $\varepsilon_1$ and $\varepsilon_2$ are selected as 0.178 and 0.207, respectively. The greedy algorithm is then implemented, with 0.85 chosen as the optimal tuning parameter among $\{0.55,0.65,0.75,0.85\}$ based on the KL divergence \eqref{eq:normKL}. Figure~\ref{fig:3} displays the absolute sample correlation matrix for the intermediate set of screened variables in Phase 1. In Phase 2, the greedy algorithm extracts 7,313 and 1,087 variables, which are organized into four diagonal quasi-bicliques.

Figure \ref{fig:2} displays the histogram of absolute sample correlations for the entire datasets $X$ and $Y$, along with the KL-divergence curves across varying values of $\epsilon$ and $\lambda$. Based on the mixture distribution modeling, the proportion of the true signal distribution is estimated as $\hat{\pi}_1 = 0.0006$ with a mean of $\hat{\mu}_1 = 0.165$. The hard-thresholding level for variable reordering ($\epsilon_1$) and the significance cutoff for binary edges ($\epsilon_2$) are selected as $0.178$ and $0.207$, respectively. As shown in the optimization plots, the tuning parameter $\lambda$ is chosen to be $0.85$ based on the maximum KL-divergence.

Figure \ref{fig:3} shows that the absolute cross correlation matrix of variables screened in Phase 1. In the bottom left of plot, there are a block diagonal structure with seventeen subgraphs. The average absolute sample correlation of intermediate screened variables in Phase 1 is 0.098, which is notably higher than the overall average of 0.047 across all variables. Within the extracted subgraphs, the average absolute sample correlation of all subgraphs is 0.223. The average absolute correlations and sizes of subgraphs are given in Table \ref{table:2}. These elevated correlation levels suggest that the selected features exhibit stronger coordinated variation than would be expected by chance, indicating potentially meaningful biological associations. In particular, the quasi-biclique structures may correspond to distinct regulatory modules in which DNA methylation changes influence the expression of specific gene clusters. Identifying such modules could provide mechanistic insights into how epigenetic dysregulation contributes to Alzheimer’s pathology, and the detected patterns may serve as candidate biomarkers for disease staging or progression risk. Furthermore, the ability to isolate these high-correlation substructures in ultra-high-dimensional settings highlights the scalability and applicability of GIDS to other complex diseases and multiomics integration tasks.

\begin{figure}[h]
\includegraphics[width=0.45\textwidth]{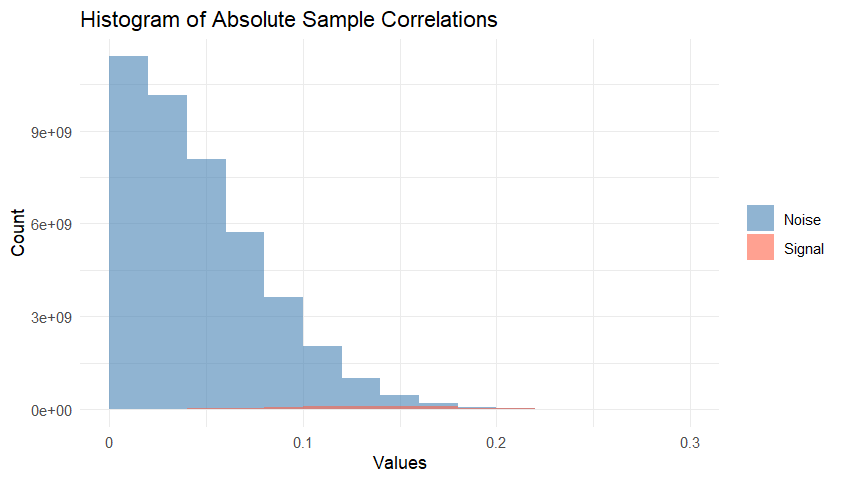}
\includegraphics[width=0.45\textwidth]{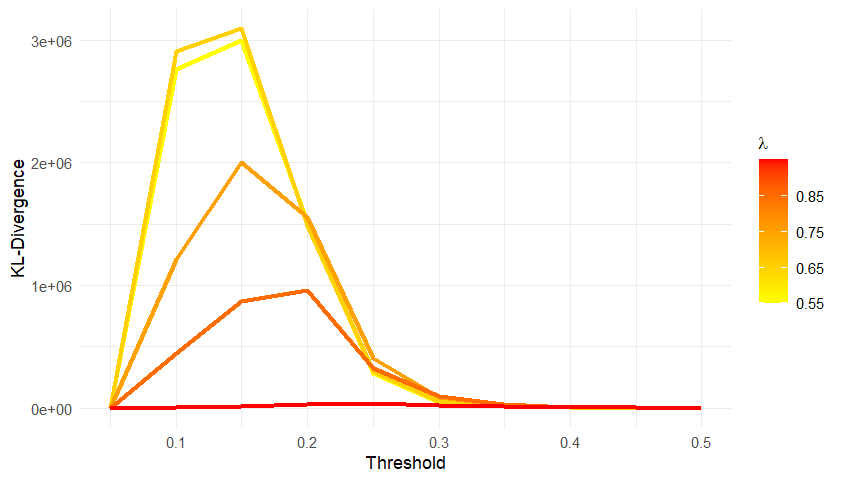}
\caption{The histgram of absolute correlations of the entire $X$ and $Y$ (left) and KL-divergence according the incremental threshold values $\epsilon$ and $\lambda$ (right) for ADNI data set. }
\label{fig:2}
\end{figure}

\begin{figure}[h]
\includegraphics[width=1\textwidth]{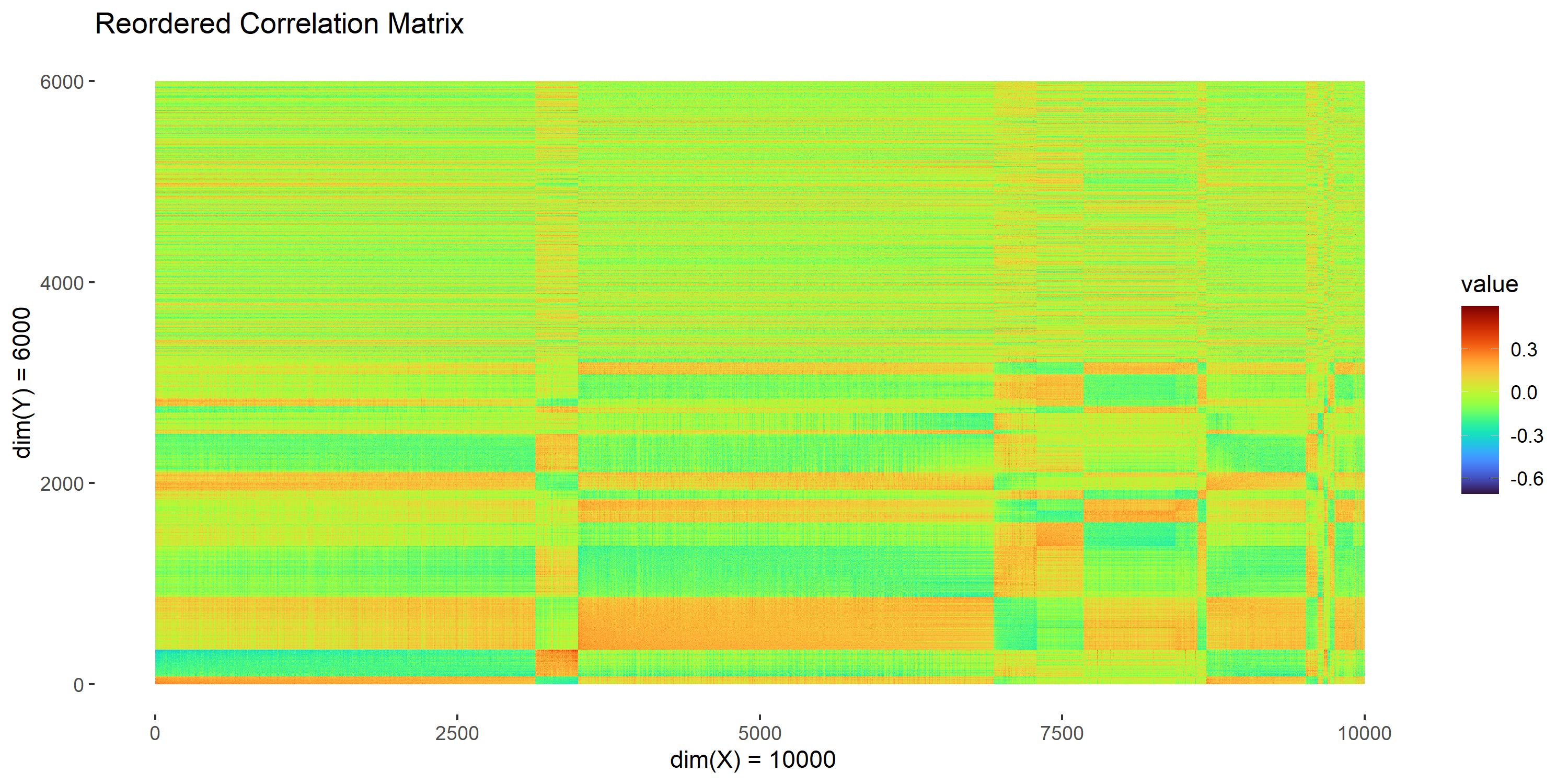}
\caption{The cross absolute correlation matrix of screened variables for Phase 1. Seventeen subgraphs are extracted in Phase 2, which are represented as block diagonal patterns. The average absolute correlation of the Phase 1 screened variables is 0.098, which is significantly higher than that of the entire variables (0.047). The average absolute correlations of extracted subgraphs are around 0.22 with the average of all subgraphs 0.223.}
\label{fig:3}
\end{figure}

\begin{table}[h]
    \caption{Profiles of seventeen extracted subgraphs from the ADNI data set.}
    \tiny
    \centering
    \setlength{\tabcolsep}{3pt} 
    \begin{tabular}{c|ccccccccccccccccc}
    \hline
        Subgraph No. & 1 & 2 & 3 & 4 & 5 & 6 & 7 & 8 & 9 & 10 & 11 & 12 & 13 & 14 & 15 & 16 & 17\\\hline
        Size of $X$ & 1883 & 1584 & 549 & 842 & 375 & 473 & 396 & 182 & 117 & 162 & 168 & 166 & 118 & 73 & 65 & 82 & 75 \\
        Size of $Y$ & 49 & 80 & 172 & 103 & 128 & 66 & 55 & 46 & 40 & 56 & 61 & 33 & 59 & 30 & 40 & 36 & 33 \\ 
        Avg-ab-correlation   & 0.267 & 0.226 & 0.210 & 0.215 & 0.207 & 0.216 & 0.206 & 0.203 & 0.209 & 0.199 & 0.200 & 0.200 & 0.198 & 0.208 & 0.198 & 0.195 & 0.197 \\\hline 
    \end{tabular}
    \label{table:2}
\end{table}

\section{Theoretical Properties of GIDS}
\label{sec:5}

\subsection{Assumptions}
\label{ssec:ass}
We introduce two assumptions underlying the theories in the GIDS  procedure. Our first assumption concerns the distributional properties of correlation coefficients in sample correlations $\{R_{ij}\}_{i\in[p],j\in[q]}$. Specifically, we assume that the sample correlation coefficients, given a sample size of $n$, satisfy a subgaussian concentration property, ensuring that they are sufficiently bounded with high probability \citep{boucheron2003concentration}. This assumption is formulated based on the fact that a sample correlation coefficient of two bivariate normal variables satisfies a subgaussian-type concentration inequality, as provided in Supplementary Materials \ref{ssec:ass1}.
\begin{assumption}[$n^{-1/2}$ concentration of sample correlation coefficients]
    Let $R_{ij}$ be the sample correlation coefficient of $X_i$ and $Y_j$ with size $n> 3$, which is generated from a bivariate distribution with a ground truth correlation coefficient $\rho_{ij} \in (-1,1)$. Then, for all $a>0$ and all $\rho_{ij} \in (-1,1)$, there are $s_1>0$ such that
    \begin{eqnarray*}
        P(|R_{ij}-\rho_{ij}| > a) \leq 2 \exp \left(-\frac{na^2}{2s_1^2}\right).
    \end{eqnarray*}
    \label{ass:sub}
\end{assumption}

Note that $R_{ij} - \rho_{ij}$ is not exactly centered. However, since $R_{ij}$ is a consistent estimator of $\rho_{ij}$, the potential bias can be addressed by adopting a slightly larger constant $s_1$. Consequently, Assumption \ref{ass:sub} is equivalent to requiring the existence of a constant $s_{12} > 0$ such that
\begin{eqnarray*}
\mathbb{E}\left[e^{a(R_{ij} - \rho_{ij})}\right] \leq 2 e^{s_{12}^2 a^2 / n}, \quad \text{for all } a > 0,
\end{eqnarray*}
in accordance with the equivalent characterizations of sub-Gaussian random variables \citep{vershynin2018high}. For convenience, we denote by $s_1$ the maximum of $s_{11}$ and $s_{12}$, where $s_{11}$ is defined as $s_{11} = \inf\left\{s > 0: \mathbb{P}(|R_{ij} - \rho_{ij}| > a) \leq 2 \exp\left(-n a^2/2 s^2\right)~\text{for~all}~a>0\right\}$.



Next, we study the behavior of the truncated absolute  sample correlation coefficient differences, which play a critical role in the row and column exclusion process of the greedy algorithm. To address the challenge of the spurious correlations caused by noise variables in ultra-high dimensional setups, we introduce the hard thresholding $\varepsilon$. More specifically, based on Assumption \ref{ass:sub}, the variance of truncated absolute sample correlation exponentially decreases in the sample size such that $$\text{Var}(|R_{ij}^{\varepsilon}|) \leq E[|R_{ij}^{\varepsilon}|^2|] \leq P(|R_{ij}^{\varepsilon}| > \varepsilon) \leq 2\exp(-n\varepsilon^2/2s_1^2).$$ 
This exponentially decreasing variance facilitates the screening in ultra-high dimensional setups, by lowering the error rate of sure screening property to $O(\exp(-n))$ from $O(\exp(-\sqrt{n})$.

At each round $t$, the algorithm selects rows and columns based on their respective aggregate scores among those in the active node sets $\widetilde{U}_{c,t}^{\lambda}$ (rows) and $\widetilde{V}_{c,t}^{\lambda}$ (columns). For the algorithm to make the correct exclusions at round $t$, it is required that

$$
\sum_{i \in \widetilde{U}_{c,t}} \left(|R_{ij}^\varepsilon| - |R_{ij'}^\varepsilon|\right) > 0
$$

for all $j \in I_Y \cap \widetilde{V}_{c,t}^{\lambda}$ and $j' \in V\backslash I_Y \cap \widetilde{V}_{c,t}^{\lambda}$, and similarly for the column sums. This left-hand side can be decomposed as

\begin{eqnarray}
\sum_{i \in \widetilde{U}_{c,t} \cap I_X} \left(|R_{ij}^\varepsilon| - |R_{ij'}^\varepsilon|\right) + \sum_{i \in \widetilde{U}_{c,t} \backslash I_X} \left(|R_{ij}^\varepsilon| - |R_{ij'}^\varepsilon|\right).\label{eq:twoterms}    
\end{eqnarray}
When $t = 1$, the first term above is likely to satisfy an anti-concentration inequality with high probability. There exists a positive constant $b$ such that
\begin{eqnarray*}
\sum_{i \in I_X} \left(|R_{ij}^\varepsilon| - |R_{ij'}^\varepsilon|\right) > b,    
\end{eqnarray*}
since all pairs $(i,j)$ are within a subgraph $B_c$ while all $(i,j')$ are outside $B_c$. To prevent incorrect exclusions (i.e., the premature exclusion of significant variables), the second term in \eqref{eq:twoterms} must remain sufficiently small, specifically less than $b$ across all $t$. To get a concentration of the second term in \eqref{eq:twoterms}, we utilize an upper bound for the moment generating function of a variable $|R_{ij}^\varepsilon| - |R_{ij'}^\varepsilon|$ such that
\begin{eqnarray*}
   E[ e^{a(|R_{ij}^\varepsilon| - |R_{ij'}^\varepsilon|)}  ] \leq 1 + 4\exp(-n\varepsilon^2/2s_1^2+|a|),
\end{eqnarray*}
for all $a>0$ and $Z_{ij}=Z_{i,j'}=0$.

Note that when the greedy algorithm decides to exclude a row or column, the aggregate scores of rows and columns are given. Given that the sample correlation coefficients $R_{ij}$ and $R_{ij'}$ are approximately independent of other sample correlation coefficients when their population correlations are all close to zero, we can account for the weak dependence between them. To correctly handle the variability in the second term under this weak dependence, it is essential to account for the information conditioned on the row aggregate scores, $\textbf{r.sum}$. Accordingly, we consider $s_2$ to quantify the extent of variance inflation, reflecting that the upper bound on the variance of sample correlations may increase to $(s_1^2+s_2^2)/n$ from $s_1^2/n$ due to this weak dependence. We construct the following assumption based on this variance inflation to ensure that errors of $|R_{ij}^\varepsilon|-|R_{ij'}^\varepsilon|$ given row sums are sufficiently concentrated.

\begin{assumption}
    Consider the exclusion process of the greedy algorithm at round $t$. If $i \in \widetilde{U}_{c,\tau}$ and $j,j' \in \widetilde{V}_{c,\tau}$ for all $\tau \leq t$ and $Z_{ij}=Z_{ij'}=0$, there is $s_2$ such that    
    \begin{eqnarray*}
        E[e^{a(|R_{ij}^\varepsilon|-|R_{ij'}^\varepsilon|)} |\sigma\{\{\textbf{r.sum}_{\tau'}\}_{\tau'=1}^{\tau}\}] \leq 1 + 4\exp(-n\varepsilon^2/2(s_1^2+s_2^2)+|a|)
    \end{eqnarray*}
    for all $\tau \leq t$ and all $a>0$.
    
    \label{ass:ass2}
\end{assumption} 

Similarly, we assume that $E[e^{a(|R_{i'j}^\varepsilon|-|R_{ij}^\varepsilon|)} |\sigma\{\{\textbf{c.sum}_{\tau'}\}_{\tau'=1}^{\tau}\}] \leq 1 + 4\exp(-n\varepsilon^2/2(s_1^2+s_2^2)+|a|)$ if $j \in \widetilde{V}_{c,\tau}$ and $i,i' \in \widetilde{U}_{c,\tau}$ for all $\tau \leq t$ and $Z_{ij}=Z_{ij'}=0$,  for all $\tau \leq t$ and all $a>0$. We show that the above assumption holds for the sample correlation coefficients of independent normal random variables in Supplementary Materials \ref{ssec:ass2}. 

\subsection{Theoretical Results}

In this subsection, we show the theoretical results of GIDS for a joint data set. First, we show our main result about the performance guarantee under the existence of only a single subgraph. Within the subgraph $B(U_1, V_1; E_1)$, the true absolute correlations $|\rho_{ij}|$ for edges $e_{ij} \in E_1$ are supposed to be $|\rho_{ij}|\geq \rho $. Here, we do not consider an assumption with a constant $\kappa$ such that $\rho \geq cn^{-\kappa}$ for $c>0$, which is a commonly chosen condition similar to Condition 3 in \cite{fan2008sure}, because our screening method does not need an assumption about the size of true positive variable sets. The result demonstrates that Phase 1 screening of GIDS can include all the true positive variables with a high probability, which is $1-O(\exp(-\Omega(n) ))$ for a sample size $n$, if the target screening size is bigger than the size of sets of true positive variables.

\begin{theorem}
    
    Under the presence of a single subgraph with the true absolute correlations lying between $[\rho,1]$ and $\delta_1=\delta_2=0$,  if the size of true submodel is less than $(p_{\text{phase1}},q_{\text{phase1}})$, a submodel $\widehat{\mathcal{M}}_{\text{phase1}}$ generated by the intermediate screended variable from Phase 1 by the greedy algorithm with the granularity $(k_1,k_2)=(1,1)$ includes the true subset of variables such that   
\begin{eqnarray*}
    \mathbb{P}( \mathcal{M}_\star \subset \widehat{\mathcal{M}}_{\text{phase1}} ) 
    >
    1-\delta(p,p_0,q,q_0,n)
\end{eqnarray*}
where $\delta(p,p_0,q,q_0,n) = O(\exp\left(-c_1p_0n + c_2p \right)+\exp\left(-c_1q_0n + c_2q \right))$ for some $c_1,~c_2>0$.
\label{thm:1}
\end{theorem}

The next result demonstrates that the extracted subgraph in Phase 2 is idential to the true submodel with high probability under the presence of single subgraph.

\begin{theorem}
    
    Under the presence of a single subgraph with the true absolute correlations lying between $[\rho,1]$ for $\rho>0$ and $\delta_1=\delta_2=0$,  there is a set of $\lambda$ in $(0,1)$ to guarantee that the greedy algorithm for GIDS has
\begin{eqnarray*}
    \mathbb{P}( \mathcal{M}_\star = \widehat{\mathcal{M}}_\lambda ) 
    >
    1-\delta(p,p_0,q,q_0,n)
\end{eqnarray*}
where $\delta(p,p_0,q,q_0,n) = O(\exp\left(-c_1'p_0n + c_2'p \right)+\exp\left(-c_1'q_0n + c_2'q \right))$ for some $c_1',~c_2'>0$ and $\widehat{\mathcal{M}}_\lambda$ is the submodel of associated pairs screened by GIDS with the tuning parameter $\lambda$. 
\label{thm:2}
\end{theorem}

The above result provides a performance guarantee of the greedy algorithm \citep{charikar2000greedy} with a sample size ($n=O(\max\{p/p_0,q/q_0\}$)) for the full recovery property under the assumptions, which are described in the previous subsection.

Next, we describe the statistical properties of the densest subgraph. After extracting each subgraph, we determine whether it is statistically significant or merely a spurious finding. This decision guides whether the algorithm proceeds to extract additional subgraphs. To formalize this process, we propose a statistical inference framework for testing the significance of each extracted dense subgraph. Specifically, we consider the following hypothesis test: 
\begin{eqnarray*}
&&H_0:\text{~the~extracted~dense~subgraph~is~not statistically~significant}\\
&&H_1:\text{~the~extracted~dense~subgraph~is~statistically~significant}
\end{eqnarray*}

We propose the following theorem as a graph-combinatorics-based procedure to determine the rejection region for the above hypothesis testing.

\begin{lemma}
    Let $B_c=(U_c,V_c;E_c)$ be a subgraph of a bipartite Erdős-Rényi $B(U, V; E)$ and binary edges with probability $\gamma_0$. Then, we have

$$
\mathbb{P}\left(\max_{U_c\subset U,V_c\subset V:|U_c|=p_c,|V_c|=q_c} 
 \frac{|E_c|}{p_cq_c}>\gamma\right) \leq \exp \left( - p_c q_c \left[ D(\gamma \| \gamma_0) - \frac{\log \left( \frac{pe}{p_c} \right)}{q_c} - \frac{\log \left( \frac{qe}{q_c} \right)}{p_c} \right] \right),
$$
where $D(a\|b) = a\log (a/b) +(1-a) \log ((1-a)/(1-b))$. 

    \label{lem:1}
\end{lemma}

This lemma is a renovated version of Lemma 2.1 in the work of \cite{chen2024identifying}. The suggested bound for the confidence level is tighter than that of the existing one. In addition, the suggested one does not require the conditions about $|U_c|$, $|V_c|$, $p$, and $q$, which are required for the existing one.

\section{Discussion}
\label{sec:6}

We have developed a new graph-based screening tool - GIDS to decipher the systematic association between two sets of high-dimensional variables. Compared to other screening methods, GIDS seeks not only to screen responses variables as well as predictors but also to identify the latent patterns of correlated sets of variables, taking the concept of a bipartite graph into account. GIDS can better differentiate the correlated variable sets based on the effective thresholding step and yield better results with minimal memory constraints. 

We provide a computationally efficient solution to implement GIDS with the upper bound of complexity $O( p+q+g(p_{\text{phase1}} + q_{\text{phase1}})^2)$.  We also show that the greedy algorithm for GIDS guarantees the sure screening property and full recovery of true related and irrelevant variables for associations between two high-dimensional datasets with an error rate of $O(\exp(-\Omega(n)))$ under mild assumptions. The numerical studies validate our theoretical results by showing that GIDS outperforms conventional methods to more accurately reveal the true submodel.


\begin{acks}[Acknowledgments]
This research was funded by the National Institute on Drug Abuse of the National Institutes of Health under Award Number 1DP1DA04896801. 
\end{acks}

\section*{Availability and Implementation}

The R code that implements GIDS is available at \href{https://github.com/hjpark0820/GIDS}{https://github.com/hjpark0820/gCCA}. The data used in this work is available at \href{https://adni.loni.usc.edu}{https://adni.loni.usc.edu}.

\bibliographystyle{biom}
\bibliography{biomsample}

\section{Supplementary Materials}
\label{sec:supp}

\subsection{Proof of Proposition \ref{prop:1}}
Note that  $\mathbb{P}\big((i,j)\in \widehat{\mathcal{M}}_j \big)\geq \delta_2$ for all $(i,j)$. Then, for all $i$, we have
$\mathbb{P}\big(i \in \widehat{\mathcal{M}}_j ~\text{for~any}~j\big)\geq 1-(1-\delta_2)^q$. We apply this to all $i$ such that $\cap_{i=1}^p \mathbb{P}\big(i \in \widehat{\mathcal{M}}_j ~\text{for~any}~j\big)= \mathbb{P}\big(|\cup_{j=1}^q\widehat{\mathcal{M}}_j|=p\big)\geq (1-(1-\delta_2)^q)^p$. As $(1-x)^r \geq 1-xr$, for $0\leq x \leq 1$ and $r\geq 1$, we have $ \mathbb{P}\big(|\cup_{j=1}^q\widehat{\mathcal{M}}_j|=p\big)\geq 1-p(1-\delta_2)^q$.

\subsection{Proof of Proposition \ref{prop:2}}

\begin{proof}
For a fixed competitor $A=(S,T)$, the condition that $B_0$ has at least as large expected $\lambda$-density as $A$ is exactly
\[
\operatorname{den}_{\lambda}(U_0,V_0)\geq \operatorname{den}_{\lambda}(S,T)
\quad\Longleftrightarrow\quad
\frac{W_0}{m_0^{\lambda}}\geq \frac{W_A}{m_A^{\lambda}}.
\]
If $m_A=m_0$, this reduces to $W_0\geq W_A$. If $m_A\neq m_0$ and $W_A>0$, taking logarithms gives
\[
\log(W_0/W_A)\geq \lambda\log(m_0/m_A).
\]
When $m_A<m_0$, the denominator $\log(m_0/m_A)$ is positive, so the inequality is equivalent to
\[
\lambda\leq \frac{\log(W_0/W_A)}{\log(m_0/m_A)}.
\]
When $m_A>m_0$, the denominator $\log(m_0/m_A)$ is negative, so the inequality reverses and becomes
\[
\lambda\geq \frac{\log(W_0/W_A)}{\log(m_0/m_A)}.
\]
Intersecting these constraints over all nonempty competitors $A\neq B_0$ and with the admissible parameter range $[1/2,1)$ gives exactly $\Lambda(B_0)$. Hence $\Lambda(B_0)\neq\varnothing$ if and only if there exists at least one admissible value of $\lambda$ for which $B_0$ maximizes the expected $\lambda$-density. 

\end{proof}

\subsection{Thresholding for Screening and Inference Binary Edges}
\label{sub:thresholding}
To infer subgraphs composed of unobserved binary edges, we leverage the absolute sample correlations $|R_{ij}|$, where a higher value indicates stronger evidence for the presence of an edge $e_{ij}$. Specifically, the conditional probability $\mathbb{P}(Z_{ij} = 0 \mid |R_{ij}| > r)$ corresponds to the (global) false discovery rate (FDR), while $\mathbb{P}(Z_{ij} = 0 \mid |R_{ij}| = r)$ corresponds to the local false discovery rate (local FDR). Both quantities are increasing functions of $r$ for proper $g_0$ and $g_1$, reflecting that higher absolute correlations are more likely to correspond to true signal edges. However, both quantities rely on accurate esimation of $\pi_1=\mathbb{P}(Z_{ij}=1)$, $g_0$ and $g_1$, which are known to be unstable in the modern statistics. This estimation challenge limits the practical reliability of FDR-based procedures for subgraph inference in complex modern datasets.


Instead, we use two types of estimators for $Z_{ij}$ utilizing a threshold $\varepsilon$: i) $|R_{ij}^\varepsilon| = |R_{ij}| I(|R_{ij}| > \varepsilon)$ and ii) $I(|R_{ij}| > \varepsilon)$ for different purposes. First, the former is a hard thresholding, which is a continuous-valued estimator after filtering. It modifies the original signal but preserves quantitative information. The latter is a classification with a cutoff. It discards magnitude and reduces everything to \textit{significant} versus \textit{insignificant}. For both quantities, the way of choosing a cutoff $\varepsilon$ depends on applications. For examples, in compressed sensing, iterative hard thresholding algorithms use a cutoff to retain only the largest entries in each update step, thereby enforcing sparsity and enabling signal recovery from limited linear measurements \cite{donoho2006compressed}, while wavelet shrinkage with hard thresholding achieves near-optimal minimax rates in nonparametric function estimation \citep{donoho1994ideal}. In addition, in classification settings, cutoffs in classification are often chosen from ROC curves, using criteria such as maximizing Youden’s index, fixing a false positive rate, minimizing distance to the ROC’s upper-left corner \citep{youden1950index}, or maximizing the F1-score to balance precision and recall \cite{fawcett2006introduction}. These examples highlight how thresholding adapts to different goals, preserving quantitative structure in estimation versus making binary decisions in classification.

First, we consider $|R_{ij}^\varepsilon|$ for the greedy algorithm and subgraph extraction. The greedy algorithm sequentially excludes rows or columns by comparing the row and column sums. For example, we consider two predictors, $X_{i_1}$ and $X_{i_2}$: the sign of the difference between their aggregated absolute correlations with all responses, $\sum_{j=1}^q |R_{i_1j}| - \sum_{j=1}^q |R_{i_2j}|$, can be used as a measure to determine which predictor is more strongly associated with the responses. However, in high-dimensional settings, spurious correlations may inflate the variability of these aggregate scores, potentially leading to misleading conclusions about the relative importance of predictors. In contrast, using the aggregate score difference $\sum_{j=1}^q |R_{i_1j}^\varepsilon| - \sum_{j=1}^q |R_{i_2j}^\varepsilon|$ for an appropriately chosen $\varepsilon > 0$, helps to suppress the contribution of weak and noisy associations. This reduction in variance leads to a higher probability of correctly identifying more informative predictors.

We use $I(|R_{ij}|>\varepsilon)$ to infer whether $|R_{ij}|$ for an individual pair $(i,j)$ arises from a signal or noise distribution. This binary estimator is straightforwardly interepreted as a proxy for the unobserved edge such that $\widehat{Z}_{ij}=I(|R_{ij}|>\varepsilon)$. The choice of $\varepsilon$ involves a trade-off between sensitivity and precision, which is a common problem in statistics \citep{fawcett2006introduction}. To balance this trade-off, we treat $\varepsilon$ as a tuning parameter and explore its impact on the recovery of the underlying subgraph structure. Specifically, by applying the thresholding rule $I(|R_{ij}| > \varepsilon)$ across all predictor-response pairs, we obtain a binary matrix that serves as a proxy for the latent adjacency matrix ${Z_{ij}}$. This binary matrix can then be analyzed to identify structured patterns, such as quasi-biclique subgraphs, which reflect coordinated associations between subsets of predictors and responses. The use of hard thresholding not only improves robustness to noise but also facilitates the identification of block structures that are otherwise obscured by spurious correlations. How to choose an appropriate $\varepsilon$ for different tasks will be introduced in the following paragraphs.

We incorporate a thresholding procedure for the row and column exclusion. At each iteration, the optimal cutoff is desirably supposed to maximize the probability of excluding rows and columns from the true negative set rather than ones from the true positive set. Since the greedy algorithm makes exclusions based on the sample means of $|R_{ij}|$ across rows or columns, the first thresholding step is designed to increase the likelihood that the sample mean corresponding to a true positive exceeds that of a true negative. Specifically, we aim to maximize the right decision probability
\begin{eqnarray*}
    \mathbb{P}\left(\left. \sum_{j\in V} |R_{ij}^\varepsilon| >  \sum_{j\in V} |R_{i'j}^\varepsilon|\right| i \in I_X~\text{and}~i'\in U\backslash I_X \right),
\end{eqnarray*}
with respect to the cutoff $\varepsilon$. For analytical tractability, we assume $p = q$ and $|U_c| = |V_c|$, as relaxing this assumption introduces unnecessary complexity. We have the closed form of the solution of the above as follows:
\begin{eqnarray}
    \varepsilon_1 = \text{argmax}_\varepsilon  \frac{  (S_1(\varepsilon) \mu_1(\varepsilon)  - S_0(\varepsilon)\mu_0(\varepsilon))^2  }{(1  + 2\delta_2/\delta) \sigma_1(\varepsilon)^2S_1(\varepsilon) + (1+2(1-\delta_2)/\delta)\sigma_0(\varepsilon)^2S_0(\varepsilon)},\label{eq:ep1}
\end{eqnarray}
where $\delta = \sqrt{(\pi_1-\delta_2)(1-\delta_1-\delta_2)}$; $\pi_1$ is the proportion of true positive $\mathbb{P}(Z_{ij}=1)$; $\mu_k(\varepsilon) = \mathbb{E}[|R_{ij}|||R_{ij}|>\varepsilon,Z_{ij}=k]$ for $k=0,1$; $\sigma_k(\varepsilon)^2= \text{Var}[|R_{ij}|||R_{ij}|>\varepsilon,Z_{ij}=k]$ for $k=0,1$; $S_k(\varepsilon) = \mathbb{P}(|R_{ij}|>\varepsilon|Z_{ij}=k)$  for $k=0,1$. In this work, we fix the threshold $\varepsilon_1$ throughout the exclusion process, as dynamically thresholding would require significantly more computation and memory.  

We now describe how to choose the cutoff of binary edges for the stopping criterion of extraction procedure by testing the significance of the $c$th extracted subgraph.  Under the null hypothesis (i.e. there are no significant associations between $X$ and $Y$), the full graph with binary edges $\{I(|R_{ij}|>\varepsilon)\}_{i,j}$ can be modeled as an Erdős–Rényi model regardless of the value of $\varepsilon$. However, selecting an appropriate threshold remains critical, as it governs the trade-off between sensitivity and precision in identifying signal edges. To optimally balance these, the threshold is selected by maximizing the F1 score such that
\begin{eqnarray}
    \varepsilon_2 = \text{argmax}_{\varepsilon\in(0,1)} \frac{2\cdot\text{precision}(\varepsilon)\cdot \text{sensitivity}(\varepsilon)}{\text{sensitivity}(\varepsilon) + \text{precision}(\varepsilon) },\label{eq:ep_bin}
\end{eqnarray}
where $\text{sensitivity}(\varepsilon)= S_1(\varepsilon)$ and $\text{precision}(\varepsilon) = \pi_1S_1(\varepsilon)/(\pi_0S_0(\varepsilon) + \pi_1 S_1(\varepsilon))$. 


In practice, as estimating $\pi_1$, $S_0$, and $S_1$ becomes computationally burdensome in ultra-high-dimensional settings, we use the discretized data (histogram-valued data) to estimate the quantities to facilitate the calculation of likelihood function. This data discretization enables fast calculations and alleviates the memory problem without significantly hindering the estimation accuracy \citep{kang2023classification}.  We employ the EM algorithm to estimate these quantities with a Fisher's Z-transformed absolute sample correlation $ 0.5\log (1+|R_{ij}|)/(1-|R_{ij}|)$, which is supposed to follow $trN(0,\sigma_0^2,a=0,b=\infty)$ and $trN(\mu_1,\sigma_1^2,a=0,b=\infty)$ for $\mu_1>0$ under $Z_{ij}=0$ and $Z_{ij}=1$, respectively, where $trN$ is a truncated normal distribution.



\subsection{Cuttoff Derivation for Hard Thresholding}

First, we describe the assumptions for the cutoff calculation. For analytical tractability, we assume $p = q=d$ and $p_0=|I_X|= q_0=|I_Y|$ under the presence of a single subgraph. Note that the proportion of signal distribution, $\pi_1$, the parameters of signal and noise distributions are calculated based on the gaussian EM algorithm\textemdash$S_k(\varepsilon)$, $\mu_k(\varepsilon)$, and $\sigma_k(\varepsilon)$ for $k=0,1$ are available for all values $\varepsilon \in \bm{\varepsilon}$.

Note that we maximize the following right decision probability
\begin{eqnarray*}
    \mathbb{P}\left(\left. \sum_{j\in V} |R_{ij}^\varepsilon| >  \sum_{j\in V} |R_{i'j}^\varepsilon|\right| i \in I_X~\text{and}~i'\in U\backslash I_X \right),
\end{eqnarray*}
representing the probability that the aggregated score of a row in the subset $I_X$ ($i\in I_X$) is larger than that of a row not in $I_X$ ($i'\notin I_X$) and thereby excluding the row not in $I_X$ at the begining of exclusion process. In this work, we do not consider a dynamic thresholding due to the calculation burden. For $|R_{ij}|$ such that $Z_{ij}=k$, we assume that the absolute correlation
\begin{eqnarray*}
    |R_{ij}| | (|R_{ij}|> \varepsilon) = \mu_k(\varepsilon) + \eta_{ij}(\varepsilon),
\end{eqnarray*}
where $\mu_k(\varepsilon) = E[|R_{ij}| | |R_{ij}| > \varepsilon]$ and $\eta_{ij}(\varepsilon)$ is a mean zero noise such that $\sigma_k(\varepsilon)^2 = \text{Var}(\eta_{ij}(\varepsilon))$.

\begin{flalign*}
&P\left(\sum_{k\in I_Y} |R_{ik}^\varepsilon| + \sum_{k\in [q]\backslash I_Y} |R_{ik}^\varepsilon|  \geq \sum_{k\in I_Y} |R_{i'k}^\varepsilon|  + \sum_{k\in [q]\backslash I_Y} |R_{i'k}^\varepsilon| \right) \\
&= P\left(\sum_{k\in I_Y} (\mu_{Z_{ik}}(\varepsilon)+\eta_{ik}(\varepsilon)) I(|R_{ik}|>\varepsilon) - \sum_{k\in I_Y} (\mu_{Z_{i'k}}(\varepsilon)+\eta_{i'k}(\varepsilon))I(|R_{i'k}|>\varepsilon) \right.\\
&+ \left.\sum_{k\in [q]\backslash I_Y}(\mu_{Z_{ik}}(\varepsilon)+\eta_{ik}(\varepsilon)) I(|R_{ik}|>\varepsilon) - \sum_{k\in [q]\backslash I_Y} (\mu_{Z_{i'k}}(\varepsilon)+\eta_{i'k}(\varepsilon))I(|R_{i'k}|>\varepsilon) \geq 0\right).
\end{flalign*}

Here, $Z_{ik}$ is a latent random variable such that $P(Z_{ik} = 1)=1-\delta_1$ for $i \in I_X~\text{and}~k\in I_Y$, $P(Z_{ik} = 1)=\delta_2$, otherwise. Let $n_1= \sum_{k\in [q]} Z_{ik}$ and $n_2= \sum_{k\in [q]} Z_{i'k}$. Then, the probability in the RHS above can be written as
\begin{eqnarray*}
    P\left( (n_1 -n_2)\mu_1(\varepsilon) + (n_2-n_1) \mu_0(\varepsilon) \geq \sum_{k\in [q]}\left( \eta_{ik}(\varepsilon)I(|R_{ik}|>\varepsilon) + \eta_{i'k}(\varepsilon) I(|R_{i'k}|>\varepsilon)\right)\right).
\end{eqnarray*}
Then, we find the variance of the noise terms on the RHS such that
\begin{flalign*}
&\text{Var}\left(\sum_{k\in [q]}\left( \eta_{ik}(\varepsilon)I(|R_{ik}|>\varepsilon) + \eta_{i'k}(\varepsilon) I(|R_{i'k}|>\varepsilon)\right)\right)\\
&=\sum_{k\in [q]}\left(E[\text{Var}(\eta_{ik}(\varepsilon)I(|R_{ik}|>\varepsilon)|I(|R_{ik}|>\varepsilon))]+\text{Var}(E[\eta_{ik}(\varepsilon)I(|R_{ik}|>\varepsilon)|I(|R_{ik}|>\varepsilon)])\right.\\
&+\left.E[\text{Var}(\eta_{i'k}(\varepsilon)I(|R_{i'k}|>\varepsilon)|I(|R_{i'k}|>\varepsilon))]+\text{Var}(E[\eta_{i'k}(\varepsilon)I(|R_{i'k}|>\varepsilon)|I(|R_{i'k}|>\varepsilon)])\right)\\
&= (p_0(1-\delta_1)+(2d-p_0)\delta_2) \sigma_1(\varepsilon)^2S_1(\varepsilon) + (p_0\delta_1+(2d-p_0)(1-\delta_2))\sigma_0(\varepsilon)^2S_0(\varepsilon)
\end{flalign*}
Using the above variance, we standardize the noise variables such that

\begin{flalign*}
&P\left( (n_1 -n_2)\mu_1(\varepsilon) + (n_2-n_1) \mu_0(\varepsilon) \geq \sum_{k\in [q]}\left( \eta_{ik}I(|R_{ik}|>\varepsilon) + \eta_{i'k}I(|R_{i'k}|>\varepsilon)\right)\right) \\
&= P\left( \frac{(n_1 -n_2)\mu_1(\varepsilon) + (n_2 -n_1) \mu_0(\varepsilon)}{\sqrt{(p_0(1-\delta_1)+(2d-p_0)\delta_2) \sigma_1(\varepsilon)^2S_1(\varepsilon) + (p_0\delta_1+(2d-p_0)(1-\delta_2))\sigma_0(\varepsilon)^2S_0(\varepsilon)} }\right.\\
&\left.\geq \frac{\sum_{k\in [q]}( \eta_{ik}I(|R_{ik}|>\varepsilon) + \eta_{i'k}I(|R_{i'k}|>\varepsilon)) }{\sqrt{(p_0(1-\delta_1)+(2d-p_0)\delta_2) \sigma_1(\varepsilon)^2S_1(\varepsilon) + (p_0\delta_1+(2d-p_0)(1-\delta_2))\sigma_0(\varepsilon)^2S_0(\varepsilon)} }\right).
\end{flalign*}
As the term in the RHS has mean 0 and variance 1, we can maximize the probability by maximizing the term in the LHS. Because $n_1$ and $n_2$ are random, we maximize the average case of the LHS by taking the expected values of them, which is
$$p_0(1-\delta_1-\delta_2)S_1(\varepsilon) \mu_1(\varepsilon)  - p_0(1-\delta_1-\delta_2)S_0(\varepsilon)\mu_0(\varepsilon).$$
Accordingly, we find $\varepsilon$ maximizing
\begin{eqnarray*}
      \frac{ p_0(1-\delta_1-\delta_2)S_1(\varepsilon) \mu_1(\varepsilon)  - p_0(1-\delta_1-\delta_2)S_0(\varepsilon)\mu_0(\varepsilon) }{\sqrt{(p_0(1-\delta_1)+(2d-p_0)\delta_2) \sigma_1(\varepsilon)^2S_1(\varepsilon) + (p_0\delta_1+(2d-p_0)(1-\delta_2))\sigma_0(\varepsilon)^2S_0(\varepsilon)}}.
\end{eqnarray*}
Note that $p_0=d\sqrt{\frac{\pi_1-\delta_2}{1-\delta_1-\delta_2}}$. Therefore, we have
\begin{eqnarray*}
    \varepsilon_1 = \text{argmax}_\varepsilon  \frac{  (S_1(\varepsilon) \mu_1(\varepsilon)  - S_0(\varepsilon)\mu_0(\varepsilon))^2  }{(1  + 2\delta_2/\delta) \sigma_1(\varepsilon)^2S_1(\varepsilon) + (1+2(1-\delta_2)/\delta)\sigma_0(\varepsilon)^2S_0(\varepsilon)},
\end{eqnarray*}
where $\delta = \sqrt{(\pi_1-\delta_2)(1-\delta_1-\delta_2)}$.

\subsection{Best Subgraph Extraction based on $\lambda$}
\label{sub:lambda}

The tuning parameter $\lambda$ has a significant effect on the extraction of subsets. Specifically, larger values of $\lambda$ tend to yield smaller and denser subsets, whereas smaller values typically result in larger, sparser selections. The KL divergence is often used to identify the best-fitting model among candidates. The KL divergence can be expressed as the difference between the cross-entropy of two distributions (target and reference) and the entropy of the target distribution. However, this measure may disproportionately emphasize the cross-entropy term, leading to inflated divergence values. To mitigate this issue, we propose the KL divergence. 

To describe the procedure, consider two distributions $P_{\lambda,\varepsilon}$ (target) and $Q_{\varepsilon}$ (reference) of $D_{ij} = I(|R_{ij}|>\varepsilon)$, which is a point estimate of binary edge $Z_{ij}$. $P_{\lambda,\varepsilon}$ is a distribution with subgraphs extracted based on the tuning parameter $\lambda$. $P_{\lambda,\varepsilon}$ divides the correlation matrix into two distinct blocks: (i) the subgraphs $B^\lambda$ extracted by the tuning parameter $\lambda$, and (ii) the area outside the subgraphs. In block (i), $D_{ij}$ is more likely to be 1, whereas it is more likely to be 0 in block (ii). Specifically, $\{D_{ij}\}_{i,j}$ are assumed to have the following distribution:
\begin{eqnarray*}
    D_{ij} \sim \text{Bernoulli}(\pi_1)&~&\text{if}~(i,j)\in B_c~\text{for~some}~c \in [C]\\
    D_{ij} \sim \text{Bernoulli}(\pi_0)&~&\text{otherwise}.
\end{eqnarray*}
In contrast, we consider a reference Bernoulli distribution $Q$ with no graph patterns between $X$ and $Y$ as follows:
\begin{eqnarray*}
    D_{ij} \sim \text{Bernoulli}(\pi)~\text{for all}~(i,j).
\end{eqnarray*}
We evaluate the KL divergence between these two distributions $D_{\mathrm{KL}}\left(P_{\lambda,\varepsilon} \| Q_{\varepsilon}\right)/H( P_{\lambda,\varepsilon}  )$ over the values in the grids of $\bm{\lambda}=(\lambda_1,\dots,\lambda_{g_1})$ and $\bm{\varepsilon} = (\varepsilon_1,\dots,\varepsilon_{g_2})$, where 
\begin{eqnarray*}
D_{\mathrm{KL}}\left(P_{\lambda,\varepsilon} \| Q_{\varepsilon}\right) &=& 
 \sum_{(i, j) \in \cup_{c=1}^{C} (U_c\bigotimes V_c)}\left( \pi_1 \log \frac{\pi_1}{\pi}+\left(1-\pi_1\right) \log \frac{\left(1-\pi_1\right)}{(1-\pi)}\right) \\
&+&  \sum_{(i, j) \notin \cup_{c=1}^{C} (U_c\bigotimes V_c)}\left( \pi_0 \log \frac{\pi_0}{\pi}+\left(1-\pi_0\right) \log \frac{\left(1-\pi_0\right)}{(1-\pi)}\right),\\
H( P_{\lambda,\varepsilon}  ) &=& \sum_{(i, j) \in \cup_{c=1}^{C} (U_c\bigotimes V_c)}\left( \pi_1 \log \frac{1}{\pi_1}+\left(1-\pi_1\right) \log \left(\frac{1}{1-\pi_1}\right)\right) \\
&+&  \sum_{(i, j) \notin \cup_{c=1}^{C} (U_c\bigotimes V_c)}\left( \pi_0 \log \frac{1}{\pi_0}+\left(1-\pi_0\right) \log \left(\frac{1}{1-\pi_0}\right)\right)
\end{eqnarray*}
where $\pi_0$ and $\pi_1$ are the sample mean of $D_{ij}$ outside and in the subgraphs, respectively; $\pi$ is the overall sample mean of $D_{ij}$. Then, the optimal tuning parameter is chosen so as to maximize the KL divergence as follows:
\begin{eqnarray*}
    \lambda^\star =  \text{argmax}_{\lambda}\text{max}_{\varepsilon} D_{\mathrm{KL}}\left(P_{\lambda,\varepsilon} \| Q_\varepsilon \right).
\end{eqnarray*}
Lastly, we get the estimates of the subgraphs $\widehat{B} = \widehat{B}^{\lambda^\star}$ and true positive sets $\widehat{I}_X = \cup_{c=1}^{C} \widehat{U}_c$ and $\widehat{I}_Y = \cup_{c=1}^{C} \widehat{V}_c$. The computational complexity of the greedy algorithm is $O((p+q)+g_1(p_{\text{phase1}}+q_{\text{phase1}}))$ when a single subgraph is present, while it increases with the number of subgraphs, reaching the worst-case complexity of $O((p+q)+g_1(p_{\text{phase1}}+q_{\text{phase1}})^2)$.

\subsection{Proof of the equivalence of Assumption 1 to $E[e^{\lambda (R_{ij}-\theta_{ij})}] \leq s_3 e^{s_4^2\lambda^2/n}$}
\label{ssec:ass11}
Let $X=\sqrt{n}(R_{ij}-\theta_{ij})/(\sqrt{2}s_2)$ for ease of presentation. Then, we have
\begin{eqnarray*}
    E\left[\left|X\right|^p\right] &=& \int_{0}^\infty P(|X|^p \geq u) du\\
    &=& \int_{0}^\infty P(|X| \geq t)pt^{p-1} dt\\
    &\leq& \int_{0}^\infty s_1 e^{-t^2} pt^{p-1} dt.
\end{eqnarray*}
By letting $t^2=s$ and then using the definition of the Gamma function, we have 
\begin{eqnarray*}
\int_{0}^\infty s_1 e^{-t^2} pt^{p-1} dt  &=& (s_1/2)p\Gamma(p/2).
\end{eqnarray*}
Then, with the Stirling approximation $\Gamma(x) \leq x^x$, we get
\begin{eqnarray*}
    E\left[\left|X\right|^p\right]\leq (s_1/2)(p/2)^{p/2}.
\end{eqnarray*}
Using $Ee^{\lambda^2 X^2} = E(1+\sum_{p=1}^\infty \frac{(\lambda^2 X^2)^p}{p!} ) \leq 1+\sum_{p=1}^\infty \frac{\lambda^{2p}E[X^{2p}]}{p!}$ and $p! \geq
(p/e)^p$, we have
\begin{eqnarray*}
    Ee^{\lambda^2X^2} \leq  1+ \sum_{p=1}^\infty \frac{(s_1/2)(2\lambda^2 p)^p}{(p/2)^p} \leq \frac{s_1}{2} \sum_{p=0}^\infty (2e\lambda^2)^p = \frac{s_1}{2} \cdot \frac{1}{1-2e \lambda^2},
\end{eqnarray*}
provided that $2e\lambda^2 < 1$. Using $1/(1-x)\leq e^{2x}$ for $x\in [0,1/2]$, we have
\begin{eqnarray*}
    Ee^{\lambda^2X^2} \leq \frac{s_1}{2} e^{4e\lambda^2},~|\lambda|\leq \frac{1}{\sqrt{2e}}.
\end{eqnarray*}

Now, we focus on $Ee^{\lambda X}$. For $|\lambda| \leq 1$, using $e^x \leq x+ e^{x^2}$ and $E[X]=0$, we have
\begin{eqnarray*}
Ee^{\lambda X} \leq E(\lambda X + e^{\lambda^2 X^2}) = Ee^{\lambda^2 X^2} \leq \frac{s_1}{2} e^{4e\lambda^2}.
\end{eqnarray*}

Next, we prove the case with $|\lambda|\geq 1$. Using $2\lambda x \leq \lambda^2 +x^2$, we have
\begin{eqnarray*}
    E e^{\lambda X} \leq e^{\lambda^2/2} Ee^{X^2/2} \leq (s_1/2) e^{\lambda^2/2} e^{2e} \leq (s_1/2) e^{2e} e^{\lambda^2},~~~\text{for}~|\lambda|\geq 1.
\end{eqnarray*}

Putting the two cases $|\lambda|\leq 1$ and $|\lambda|\geq 1$ together, for $\lambda\in \mathbb{R}^1$,  we have
\begin{eqnarray*}
     E e^{\lambda X} \leq (s_1/2)e^{2e} e^{4e\lambda^2}.
\end{eqnarray*}

Letting $X=\sqrt{n}(R_{ij}-\theta_{ij})/(\sqrt{2}s_2)$, we have
\begin{eqnarray*}
     E e^{\lambda (R_{ij}-\theta_{ij})} \leq (s_1/2)e^{2e} e^{2es_1^2 \lambda^2/n}.
\end{eqnarray*}

Thus, by defining $s_3=(s_1/2)e^{2e}$ and $s_4^2 = 2e s_1^2$, we have
\begin{eqnarray}
     E e^{\lambda (R_{ij}-\theta_{ij})} \leq s_3 e^{s_4^2 \lambda^2/n}. \label{eq:ass2}
\end{eqnarray}

Now, we show that there exist $s_1,s_2>0$ such that $P(|R_{ij}-\theta_{ij}|>t) \leq s_1\exp(-\frac{t^2}{2s_1^2})$, if \eqref{eq:ass2} is true. Note that
\begin{eqnarray*}
    P(R_{ij}-\theta_{ij}\geq t) = P(e^{\lambda (R_{ij}-\theta_{ij})} \geq e^{\lambda t}) \leq e^{-\lambda t}Ee^{\lambda (R_{ij}-\theta_{ij})} \leq e^{-\lambda t}s_3e^{s_4^2 \lambda^2 }= s_3e^{-\lambda t + s_4^2 \lambda^2}.
\end{eqnarray*}
Thus, by letting $\lambda=t/2s_4^2$, we have
\begin{eqnarray*}
    P(R_{ij}-\theta_{ij} \geq t)\leq s_3 e^{-t^2/4s_4^2}.
\end{eqnarray*}
We can get the same result for $P(R_{ij}-\theta_{ij} < -t)$ for $t<0$ with the same logic. Therefore, we have
\begin{eqnarray*}
    P(|R_{ij}-\theta_{ij}| \geq t)\leq 2s_3 e^{-nt^2/4s_4^2}.
\end{eqnarray*}

\subsection{Proof of the statement in Assumption 1 for Gaussian distributions}
\label{ssec:ass1}
Let $f_n(r|\rho)$ be the density for a sample correlation $R$ of two normally distributed random variables with correlation coefficient $\rho$ and sample size $n>3$. In the work of \cite{hotelling1953new}, it is written as follows:
\begin{eqnarray*}
    f_n(r|\rho) = \frac{n-1}{\sqrt{2\pi}} \frac{\Gamma(n)}{\Gamma(n+\frac{1}{2})} (1-\rho^2)^{\frac{n}{2}}(1-r^2)^{\frac{n-3}{2}}(1-\rho r)^{-n+\frac{1}{2}} \prescript{~}{2}F_1\left(\frac{1}{2},\frac{1}{2},n+\frac{1}{2},\frac{1+r\rho}{2}\right),
\end{eqnarray*}
where $\prescript{}{2}F_1$ is the hypergeometric function such that
\begin{eqnarray*}
\prescript{}{2}F_1(a,b,c,x) =  1+ \frac{ab}{c}x + \frac{a(a+1)b(b+1)}{2!c(c+1)}x^2 + \dots,  
\end{eqnarray*}
which is convergent for $x\in(-1,1)$. Without loss of generality, we can consider $0\leq \rho<1$ as the argument for the case $\rho < 0$ is symmetric for the case $\rho >0$. For $|\rho|=1$, the upper bound is trivial. To find an upper bound of $P(|R-\rho|>\epsilon)$ for $\epsilon\geq 0$, we consider three mutually exclusive cases: (i) $R-\rho>\epsilon$, (ii) $-(R-\rho)<\epsilon$ and $0<\epsilon<\rho$, and (iii) $-(R-\rho)<\epsilon$ and $\rho<\epsilon$. First, we evaluate $P(R-\rho>\epsilon)$ for case (i). 
\begin{eqnarray*}
    &&P(R-\rho>\epsilon) \\
    &=& \int_{\rho+\epsilon}^1 \frac{n-1}{\sqrt{2\pi}} \frac{\Gamma(n)}{\Gamma(n+\frac{1}{2})} (1-\rho^2)^{\frac{n}{2}}(1-r^2)^{\frac{n-3}{2}}(1-\rho r)^{-n+\frac{1}{2}}\prescript{}{2}F_1\left(\frac{1}{2},\frac{1}{2},n+\frac{1}{2},\frac{1+r\rho}{2}\right) dr\\
    &\leq & M_1\frac{n-1}{\sqrt{2\pi}} \frac{\Gamma(n)}{\Gamma(n+\frac{1}{2})}  \int_{\rho+\epsilon}^1 (1-\rho^2)^{\frac{n}{2}}(1-r^2)^{\frac{n-3}{2}}(1-\rho r)^{-n+\frac{1}{2}} dr
\end{eqnarray*}
where $M_1 = \sup_{r\in[-1,1]}\prescript{}{2}F_1(1/2,1/2,n+1/2,\frac{1+r\rho}{2})< \infty$. Note that $\prescript{}{2}F_1\left(\frac{1}{2},\frac{1}{2},n+\frac{1}{2},\frac{1+r\rho}{2}\right) $ is positive as $0<\frac{1+r\rho}{2} < 1$ and $\prescript{}{2}F_1\left(\frac{1}{2},\frac{1}{2},n+\frac{1}{2},\frac{1+r\rho}{2}\right) \leq \prescript{}{2}F_1\left(\frac{1}{2},\frac{1}{2},\frac{3}{2},\frac{1+r\rho}{2}\right) = \text{arcsin}\left(\sqrt{\frac{1+r\rho}{2}}\right)/\sqrt{\frac{1+r\rho}{2}}$ \citep{lozier2003nist}. As $\text{arcsin}\left(\sqrt{\frac{1+r\rho}{2}}\right)/\sqrt{\frac{1+r\rho}{2}}$ is continuous in $r\in[-1,1]$, $\prescript{}{2}F_1\left(\frac{1}{2},\frac{1}{2},\frac{3}{2},\frac{1+r\rho}{2}\right)$ is bounded by a positive constant $M_1$.

Using $(a+b)/2 \geq \sqrt{ab}$ for $a,b>0$, we have
\begin{eqnarray*}
(1-\rho^2)^{\frac{n-3}{2}}(1-r^2)^{\frac{n-3}{2}}(1-\rho r)^{-(n-3)}
    &=& \left(\left(\frac{1-\rho^2}{1-\rho r}\right)^{1/2}\left(\frac{1-r^2}{1-\rho r}\right)^{1/2}\right)^{n-3}\\
    &\leq& \left(\frac{1}{2}\cdot\frac{(1-\rho^2) + 1-r^2}{1-\rho r}\right)^{n-3} \\
    &=&  \left(1 - \frac{1}{2}\cdot\frac{(r-\rho)^2}{1-\rho r}\right)^{n-3}.
\end{eqnarray*}

Accordingly, we have
\begin{eqnarray*}
    && \int_{\rho+\epsilon}^1 (1-\rho^2)^{\frac{n}{2}}(1-r^2)^{\frac{n-3}{2}}(1-\rho r)^{-n+\frac{1}{2}} dr\\
    &\leq& (1-\rho^2)^{\frac{3}{2}}(1-\rho)^{-\frac{5}{2}} \int_{\rho+\epsilon}^1 (1-\rho^2)^{\frac{n-3}{2}}(1-r^2)^{\frac{n-3}{2}} (1-\rho r)^{-n+3}  dr\\
    &\leq& (1-\rho^2)^{\frac{3}{2}}(1-\rho)^{-\frac{5}{2}} \int_{\rho+\epsilon}^1 \left(1 - \frac{1}{2}\cdot\frac{(r-\rho)^2}{1-\rho r}\right)^{n-3} dr\\
\end{eqnarray*}

Using $1-x \leq e^{-x}$ for $x>0$, for $r>\rho+\epsilon$, we have
\begin{eqnarray*}
    \left(1 - \frac{1}{2}\cdot\frac{(r-\rho)^2}{1-\rho r}\right)^{n-3} \leq \exp\left(-\frac{(n-3)(r-\rho)^2}{2(1-\rho r)}\right)\leq \exp\left(-\frac{(n-3)(r-\rho)^2}{2(1-\rho^2 )}\right).
\end{eqnarray*}

Thus, we have
\begin{eqnarray*}
    \int_{\rho+\epsilon}^1 \left(1 - \frac{1}{2}\cdot\frac{(r-\rho)^2}{1-\rho r}\right)^{n-3} dr \leq \int_{\rho+\epsilon}^1 \exp\left(-\frac{(n-3)(r-\rho)^2}{2(1-\rho^2)}\right) dr \leq \int_{\epsilon}^\infty \exp\left(-\frac{(n-3)r^2}{2(1-\rho^2)}\right) dr.
\end{eqnarray*}

Using the inequality  $1-\Phi(z)\leq \exp(-z^2/2),~z>0$, where $\Phi(\cdot)$ is the CDF of standard normal distribution, we have 
\begin{eqnarray*}
    1-\Phi\left(\frac{\sqrt{n-3}\epsilon}{\sqrt{1-\rho^2}}\right) = \int_{\epsilon}^\infty \frac{\sqrt{n-3}}{\sqrt{2\pi (1-\rho^2)}} \exp\left(-\frac{(n-3)r^2}{2(1-\rho^2)}\right) dr \leq \exp\left(-\frac{(n-3)\epsilon^2}{2(1-\rho^2)} \right).
\end{eqnarray*}

\begin{eqnarray}
    P(R-\rho>\epsilon) &\leq&  M_1 \frac{n-1}{\sqrt{2\pi}} \frac{\Gamma(n)}{\Gamma(n+\frac{1}{2})} \frac{\sqrt{2\pi (1-\rho^2)}}{\sqrt{n-3}} (1-\rho^2)^{\frac{3}{2}}(1-\rho)^{-\frac{5}{2}}\exp\left(-\frac{n\epsilon^2}{2(1-\rho^2)} \right)\nonumber\\
    &\leq& \frac{(n-1)\Gamma(n)}{\sqrt{n-3}\Gamma(n+\frac{1}{2})}  (1-\rho^2)^{2}(1-\rho)^{-\frac{5}{2}}\exp\left(-\frac{(n-3)\epsilon^2}{2(1-\rho^2)} \right)\nonumber\\
    &\leq&  M_1  M_2 (1+\rho)^2(1-\rho)^{-\frac{1}{2}}\exp\left(-\frac{(n-3)\epsilon^2}{2(1-\rho^2)} \right),\label{eq:case1}
\end{eqnarray}
where $M_2 = \sup_{n\geq 4} \frac{(n-1)\Gamma(n)}{\sqrt{n-3}\Gamma(n+\frac{1}{2})}$. $M_2$ is finite as $\frac{\Gamma(n)}{\Gamma(n+\frac{1}{2})} \leq \sqrt{\frac{1}{n+1/4}}$ as shown in \citep{qi2010bounds}.

Next, we calculate $P( -(R - \rho) >  \epsilon)$ for $0<\epsilon<\rho$ for case (ii). Note that
\begin{eqnarray*}
    && \int_{0}^{\rho-\epsilon} f_n(r|\rho) dr \\
    &\leq& \int_{0}^{\rho-\epsilon} \frac{n-1}{\sqrt{2\pi}} \frac{\Gamma(n)}{\Gamma(n+\frac{1}{2})} (1-\rho^2)^{\frac{n}{2}}(1-r^2)^{\frac{n-3}{2}}(1-\rho r)^{-n+\frac{1}{2}}\prescript{}{2}F_1\left(\frac{1}{2},\frac{1}{2},n+\frac{1}{2},\frac{1+r\rho}{2}\right) dr.
\end{eqnarray*}

Similarly to the proof for case (i), we have
\begin{eqnarray*}
    && \int_{0}^{\rho-\epsilon} (1-\rho^2)^{\frac{n}{2}}(1-r^2)^{\frac{n-3}{2}}(1-\rho r)^{-n+\frac{1}{2}} dr\\
    &\leq& (1-\rho^2)^{\frac{3}{2}}(1-\rho^2)^{-\frac{5}{2}} \int_{0}^{\rho-\epsilon} (1-\rho^2)^{\frac{n-3}{2}}(1-r^2)^{\frac{n-3}{2}} (1-\rho r)^{-n+3}  dr\\
    &\leq& (1-\rho^2)^{-1} \int_{0}^{\rho-\epsilon} \left(1 - \frac{1}{2}\cdot\frac{(r-\rho)^2}{1-\rho r}\right)^{n-3} dr\\
    &\leq& (1-\rho^2)^{-1} \int_{0}^{\rho-\epsilon} \left(1 - \frac{1}{2}(r-\rho)^2\right)^{n-3} dr\\
    &\leq& (1-\rho^2)^{-1} \int_{0}^{\rho-\epsilon} \exp\left(-\frac{(n-3)(r-\rho)^2}{2}\right) dr
\end{eqnarray*}

Accordingly, we have

\begin{eqnarray*}
    1-\Phi\left(\sqrt{n-3}\epsilon\right) = \int_{-\infty}^{\rho-\epsilon} \frac{\sqrt{n-3}}{\sqrt{2\pi}} \exp\left(-\frac{(n-3)(r-\rho)^2}{2}\right) dr \leq \exp\left(-\frac{(n-3)\epsilon^2}{2} \right).
\end{eqnarray*}

Thus, we have

\begin{eqnarray}
    &&\int_{0}^{\rho-\epsilon} \frac{n-1}{\sqrt{2\pi}} \frac{\Gamma(n)}{\Gamma(n+\frac{1}{2})} (1-\rho^2)^{\frac{n}{2}}(1-r^2)^{\frac{n-3}{2}}(1-\rho r)^{-n+\frac{1}{2}} \prescript{}{2}F_1\left(\frac{1}{2},\frac{1}{2},n+\frac{1}{2},\frac{1+r\rho}{2}\right)  dr\nonumber \\
    &\leq&  M_1 \frac{n-1}{\sqrt{2\pi}} \frac{\Gamma(n)}{\Gamma(n+\frac{1}{2})} \frac{\sqrt{2\pi}}{\sqrt{n-3}} (1-\rho^2)^{-1}\exp\left(-\frac{(n-3)\epsilon^2}{2} \right)\nonumber\\
    &\leq&  M_1 M_2 (1-\rho^2)^{-1}\exp\left(-\frac{(n-3)\epsilon^2}{2} \right).\label{eq:case2}
\end{eqnarray}

Lastly, we calculate $P( R - \rho <  - \epsilon)$ for $\rho<\epsilon$. Based on a similar logic to those in the previous two cases, we have
\begin{eqnarray*}
    && \int_{-1}^{\rho-\epsilon} (1-\rho^2)^{\frac{n}{2}}(1-r^2)^{\frac{n-3}{2}}(1-\rho r)^{-n+\frac{1}{2}} dr\\
    &\leq& (1-\rho^2)^{\frac{3}{2}} \int_{-1}^{\rho-\epsilon} (1-\rho^2)^{\frac{n-3}{2}}(1-r^2)^{\frac{n-3}{2}} (1-\rho r)^{-n+3}  dr\\
    &\leq& (1-\rho^2)^{\frac{3}{2}} \int_{-1}^{\rho-\epsilon} \left(1 - \frac{1}{2}\cdot\frac{(r-\rho)^2}{1-\rho r}\right)^{n-3} dr\\
    &\leq& (1-\rho^2)^{\frac{3}{2}}  \int_{-1}^{\rho-\epsilon} \left(1 - \frac{(r-\rho)^2}{2(1+\rho)}\right)^{n-3} dr\\
    &\leq& (1-\rho^2)^{\frac{3}{2}}\int_{-1}^{\rho-\epsilon} \exp\left(-\frac{(n-3)(r-\rho)^2}{2(1+\rho)}\right) dr
\end{eqnarray*}

Accordingly, we have

\begin{eqnarray*}
    1-\Phi\left(\frac{\sqrt{n-3}\epsilon}{\sqrt{1+\rho}}\right) = \int_{-\infty}^{\rho-\epsilon} \frac{\sqrt{n-3}}{\sqrt{2\pi(1+\rho)}} \exp\left(-\frac{(n-3)(r-\rho)^2}{2(1+\rho)}\right) dr \leq \exp\left(-\frac{(n-3)\epsilon^2}{2(1+\rho)} \right).
\end{eqnarray*}

Thus, we have

\begin{eqnarray}
    &&\int_{-1}^{\rho-\epsilon} \frac{n-1}{\sqrt{2\pi}} \frac{\Gamma(n)}{\Gamma(n+\frac{1}{2})} (1-\rho^2)^{\frac{n}{2}}(1-r^2)^{\frac{n-3}{2}}(1-\rho r)^{-n+\frac{1}{2}} \prescript{}{2}F_1\left(\frac{1}{2},\frac{1}{2},n+\frac{1}{2},\frac{1+r\rho}{2}\right)   dr \nonumber\\
    &\leq&  M_1 \frac{n-1}{\sqrt{2\pi}} \frac{\Gamma(n)}{\Gamma(n+\frac{1}{2})} \frac{\sqrt{2\pi(1+\rho)}}{\sqrt{n-3}} (1-\rho^2)^{\frac{3}{2}} \exp\left(-\frac{(n-3)\epsilon^2}{2(1+\rho)} \right)\nonumber\\
    &\leq&  M_1 M_2  (1+\rho)^{\frac{1}{2}}\exp\left(-\frac{(n-3)\epsilon^2}{2(1+\rho)} \right).\label{eq:case3}
\end{eqnarray}

Putting \eqref{eq:case1}, \eqref{eq:case2}, and \eqref{eq:case3} all together, we have
\begin{eqnarray*}
    P(|R-\rho|>\epsilon ) \leq 2 M_1 M_2(1-\rho)^{-1}  \exp\left(-\frac{(n-3)\epsilon^2}{2(1+\rho)} \right).
\end{eqnarray*}
Therefore, there exist $s_1'>0$ and $s_2'>0$ such that 
\begin{eqnarray*}
     P(|R-\rho|>\epsilon )  \leq s_1' \exp\left(-\frac{n\epsilon^2}{2s_2'} \right).
\end{eqnarray*}

\subsection{Proof of the statement in Assumption 2 for Gaussian distributions}
\label{ssec:ass2}
Let ${\bf X}_i = (X_{i1},X_{i2},\dots,X_{in})^\top$ be a random vector such that each $X_{ij}$ is independently generated from $N(\mu_i,\sigma_i^2)$. We assume that ${\bf X}_1$, ${\bf X}_2$, and ${\bf X}_3$ are mutually independent. We let centered and standardized vectors of ${\bf X}_i$, ${\bf X}_2$, and ${\bf X}_3$, with sample size $n$, denoted by ${\bf Y}_1$, ${\bf Y}_2$, and ${\bf Y}_3$, respectively. Since they are centered and standardized, they are in $\mathbb{S}^{d-1}$. In addition, $R_{ij} = {\bf Y}_i^\top {\bf Y}_j$. Because each random variable ${\bf X}_i$ is independent from each other, ${\bf Y}_i$ is unformly distributed in $\mathbb{S}^{d-1}$. Consider the joint density of ${\bf Y}_1$ and $R_{12}=r_{12}$. Then, we have
\begin{eqnarray*}
    f({\bf Y}_1=y_1,R_{12}=r_{12})=f(R_{12}=r_{12}|{\bf Y}_1=y_1)f({\bf Y}_1=y_1).
\end{eqnarray*}
Here, note that $R_{12}=r_{12}$ represents ${\bf Y}_2^\top y_1 = r_{12}$ and that $r_{12}$ represents $\text{cos}(\theta_{12})$, where $\theta_{12}$ is the angle between ${\bf Y}_1$ and ${\bf Y}_2$. The set $\{y:y^\top y_1 = r_{12}\} \in \mathbb{S}^{d-1}$ is the collection of vectors in $\mathbb{S}^{d-1}$ that have the angle $\theta_{12} = \text{cos}^{-1}(r_{12})$ with $y_1$. Since ${\bf Y}_2$ is independent of ${\bf Y}_1$ and uniformly distributed over $\mathbb{S}^{d-1}$, the probability density of $\{y:y^\top y_1 = r_{12}\}$ does not depend on the value of $y_1$. This means 
\begin{eqnarray*}
    f(R_{12}=r_{12}|{\bf Y}_1=y_1) = f(R_{12}=r_{12}).
\end{eqnarray*}
Accordingly, we have
\begin{eqnarray*}
    f({\bf Y}_1=y_1,R_{12}=r_{12})=f(R_{12}=r_{12})f({\bf Y}_1=y_1).
\end{eqnarray*}
Thus, ${\bf Y}_1$ and $R_{12}$ are indenpendent. Because ${\bf Y}_3$ is also independent from $R_{12}$, $R_{13}={\bf Y}_1^\top {\bf Y}_3$ and $R_{12}$ are independent.

\subsection{Proof of Theorem \ref{lem:1}}

\begin{proof}
    
Given a bipartite Erdős-Rényi random graph $B(U, V ; E)$ with node sizes $p$ and $q$ and binary edges with probability $\gamma_0$, we want to estimate the probability that there exists a subgraph $B_c=(U_c, V_c ; E_c)$ with $p_c$ and $q_c$ nodes in the respective parts and edge density greater than $\gamma_0$. For a fixed subset $U_c \subset [q]$ with $|U_c| = p_c$ and $V_c \subset V$ with $|V_c| = q_c$, the probability that the induced subgraph on $U_c$ and $V_c$ has density at least $q$ is given by

$$
P\left(\frac{|E_c|}{|U_c||V_c|} > \gamma \right) = \sum_{l > \gamma p_c q_c} \binom{p_c q_c}{l} \gamma_0^l (1-\gamma_0)^{p_c q_c - l}.
$$

Using a Chernoff bound for the binomial distribution, this probability can be bounded as follows:

$$
P\left(\frac{|E_c|}{|U_c||V_c|} > \gamma \right) \leq \exp \left( -D(\gamma \| \gamma_0) \cdot p_c q_c \right),
$$

where $D(\gamma \| \gamma_0)$ is the Kullback-Leibler divergence given by

$$
D(\gamma \| \gamma_0) = \gamma \log \left(\frac{\gamma}{\gamma_0}\right) + (1-\gamma) \log \left(\frac{1-\gamma}{1-\gamma_0}\right).
$$

Now, using the union bound by summing over all possible subgraph pairs, we have

$$
P\left( \max_{U_c\subset U,V_c\subset V:|U_c|=p_c,|V_c|=q_c} 
 \frac{|E_c|}{|U_c||V_c|} > \gamma\right) \leq \binom{p}{p_c} \binom{q}{q_c} \exp \left( -D(\gamma \| \gamma_c) \cdot p_c q_c \right).
$$

Using Stirling's approximation $(n_0/e)^{n_0} \leq n_0!$ and $n(n-1)\cdots (n-n_0+1) \leq n^{n_0}$ for integers $n_0<n$ and simplifying,

$$
\binom{p}{p_c} \binom{q}{q_c} \leq \exp \left( p_c \log \left( \frac{pe}{p_c} \right) + q_c \log \left( \frac{qe}{q_c} \right) \right),
$$

we obtain

$$
P\left(\max_{U_c\subset U,V_c\subset V:|U_c|=p_c,|V_c|=q_c} 
 \frac{|E_c|}{|U_c||V_c|}>\gamma\right) \leq \exp \left( p_c \log \left( \frac{pe}{p_c} \right) + q_c \log \left( \frac{qe}{q_c} \right) - D(\gamma \| \gamma_0) \cdot p_c q_c \right).
$$

This can be rewritten as

$$
P\left(\max_{U_c\subset U,V_c\subset V:|U_c|=p_c,|V_c|=q_c} 
 \frac{|E_c|}{|U_c||V_c|}>\gamma\right) \leq \exp \left( - p_c q_c \left[ D(\gamma \| \gamma_0) - \frac{\log \left( \frac{pe}{p_c} \right)}{q_c} - \frac{\log \left( \frac{qe}{q_c} \right)}{p_c} \right] \right).
$$
\end{proof}

\subsection{Proof of Theorem \ref{thm:1}}

In this subsection, we present the underlying theory for the main results of the GIDS procedure. We consider that the correlation matrix has a network structure with a single biclique subgraph $B=(U_1,V_1,E_1)$, where the absolute values of correlations $|\rho_{ij}|$ for all $(i,j) \in U_1\otimes V_1$ have the values in $[\rho,w\rho]$ for some $0<\rho<1$ and $w\geq 1$. In addition, the correlation coefficients outside the subgraph are zero. For ease of presentation, we set $I_X=\{1,2,\dots,p_0\}$ and $I_Y=\{1,2,\dots,q_0\}$. The following lemma provides bounds to ensure that the sample correlation coefficients, both within and outside the subgraph, lie within specific intervals based on the subGaussianity of sample correlation coefficients. 

\begin{lemma}
    Let $R_{ij}$ and $\rho_{ij}$ be the sample correlation coefficient with sample size $n$ and ground truth correlation of $X_i$ and $Y_j$, respectively. Then, for all $i \in [p]$ and $j \in [q]$, for $a>0$, we have
\begin{eqnarray*}
   |R_{ij}-\rho_{ij}| < a,
\end{eqnarray*}
with probability at least $1-2pq \exp \left(-na^2/2s_1^2\right)$.
\label{lem:2}
\end{lemma}

\begin{proof}
By Assumption \ref{ass:sub}, we have
 \begin{eqnarray*}
    P(|R_{ij}-\rho_{ij}| > a) \leq 2 \exp \left(-\frac{na^2}{2s_1^2}\right).
\end{eqnarray*}
By applying the union bound, we have
\begin{eqnarray*}
     P\left(|R_{ij}-\rho_{ij}| > a,~\text{for~any}~(i,j)\in [p]\otimes [q]\right) \leq 2pq \exp \left(-\frac{na^2}{2s_1^2}\right).
\end{eqnarray*}
Thus, for all $i \in [p]$ and $j \in [q]$, with probability at least $1-2pq \exp \left(-na^2/2s_1^2\right)$, we have
\begin{eqnarray*}
   |R_{ij}-\rho_{ij}| < a.
\end{eqnarray*}
\end{proof}

By Lemma \ref{lem:2}, the sample correlations are concentrated around their true means when the sample size is sufficiently large. To quantify this concentration, we introduce a positive constant $\eta > 2$. Specifically, if $|R_{ij}| > (\eta - 1) \rho / \eta$ for $(i,j)$ such that $e_{ij}\in E$ and $|R_{ij}| < \rho / \eta$ for $(i,j)$ such that $e_{ij}\notin E$, then the thresholding cutoff $\varepsilon$ in \eqref{eq:ep1} must lie within the interval $(\rho / \eta, (\eta - 1) \rho / \eta)$. For clarity and consistency in the subsequent analysis, we set $\varepsilon = \rho / 2$ in \eqref{eq:ep1}, allowing for a more straightforward presentation of the main results.

The next lemma is an anti-concentration inequality for the difference of sums of absolute sample correlation coefficients between one in a subgraph and another outside it. This inequality makes the two sums (one in a subgraph, another outside it) distinct, so that we can exclude rows or columns not in subgraphs.

\begin{lemma}
For $i_1\in I_X$, $i_2\notin I_X$, and $\eta>3$, we have $\sum_{j\in I_Y} (|R_{i_1j}^\varepsilon|  - |R_{i_2j}^\varepsilon|) > q_0\rho(\eta-2)/\eta$,
and for $j_1\in I_Y$ and $j_2\in V\backslash I_Y$
$\sum_{i\in I_X} (|R_{ij_1}^\varepsilon|  - |R_{ij_2}^\varepsilon|) > p_0\rho(\eta-2)/\eta$, with probability at least $1-\delta_0$, where $\delta_0=2pq \exp \left(-n\rho^2/2\eta^2 s_1^2\right).$ 
\label{lem:3}
\end{lemma}

\begin{proof}
    By Lemma \ref{lem:2}, for all $i$ and $j$, we have $|R_{ij}-\rho_{ij}| < a$ with probability at least $1-2pq \exp \left(-na^2/2s_1^2\right)$. By plugging $1/\eta$ to $a$, we have $|R_{ij}|> \frac{\eta-1}{\eta}\rho$ for $i\in I_{X}~\text{and}~j\in I_{Y}$ and $|R_{ij}| < \frac{1}{\eta}\rho$ for $i\in I_{X}^c~\text{or}~j\in I_{Y}^c$ with probabilty at least $1-2pq \exp \left(-n\rho^2/2\eta^2 s_1^2\right)$. Now, we consider $|R_{ij}^\varepsilon|$ for $0\leq \varepsilon < \rho/2$. As $\frac{\eta-1}{\eta}\rho>\varepsilon$ and $|R_{ij}|\geq |R_{ij}^\varepsilon|$, we still have 
\begin{eqnarray}
    |R_{ij}^\varepsilon|> \frac{\eta-1}{\eta}\rho,~~~\text{for}~i\in I_{X}~\text{and}~j\in I_{Y}\label{eq:r23}
\end{eqnarray}
and
\begin{eqnarray}
    |R_{i'j}^\varepsilon| < \frac{1}{\eta}\rho,~~~\text{for}~i'\in I_{X}^c~\text{or}~j\in I_{Y}^c.\label{eq:r13}
\end{eqnarray}

Thus, we have $\sum_{j\in I_Y} |R_{ij}^\varepsilon| > q_0((\eta-1)/\eta)\rho$ for all $i\in I_X$ and $\sum_{j\in I_Y} |R_{i'j}^\varepsilon| < q_0(\rho/\eta)$ for all $i'\in U\backslash I_X$ with probability at least $1-2pq \exp \left(-n\rho^2/2\eta^2 s_1^2\right)$. Accordingly, we have $\sum_{j\in I_Y} (|R_{ij}^\varepsilon|  - |R_{i'j}^\varepsilon|) > q_0\rho(\eta-2)/\eta$ for $i\in I_X$ and $i'\in U\backslash I_X$  with probability at least $1-2pq \exp \left(-n\rho^2/2\eta^2 s_1^2\right)$. In the same way, we can show
$\sum_{i\in I_X} (|R_{ij}^\varepsilon|  - |R_{ij'}^\varepsilon|) > p_0\rho(\eta-2)/\eta$ for $j\in I_Y$ and $j'\in V\backslash I_Y$ with probability at least $1-2pq \exp \left(-n\rho^2/2\eta^2 s_1^2\right)$.
\end{proof}

The following lemma presents a concentration inequality for the sum of differences of the column means of $|R_{ij}^\varepsilon|$ between two columns. 

\begin{lemma}

For all $t$, with probability at least $1-\delta_1$, we have
    $$\sum_{i \in \widetilde{U}_{c,t}\backslash I_X} (|R_{ij_1}^\varepsilon| - |R_{ij_2}^\varepsilon| ) \leq \frac{p_0(\eta-2)\rho}{\eta},$$
    for all $j_1,j_2\in [q]$ and 
$$\sum_{j \in \widetilde{V}_{1,t}\backslash I_Y} (|R_{i_1j}^\varepsilon| - |R_{i_2j}^\varepsilon| ) \leq \frac{q_0(\eta-2)\rho}{\eta}.$$
for all $i_1,i_2\in[p]$ and all $t=1,2,\dots,p+q-1$, where $\delta_1 = \exp\left(-\frac{p_0(\eta-2)\rho}{2\eta}\frac{n\varepsilon^2}{2(s_1^2+s_2^2)} + p\log 5 + 2\log q\right) + \exp\left(-\frac{q_0(\eta-2)\rho}{2\eta}\frac{n\varepsilon^2}{2(s_1^2+s_2^2)} + q\log 5 + 2\log p\right)$.

\label{lem:4}
\end{lemma}

\begin{proof}

Let $\{u_1,u_2,\dots,u_{p-p_0}\}$ be the sequence of excluded rows in $U\backslash I_X$ and $n_X(t)$ be the number of excluded rows in $U\backslash I_X$ by time $t$. Then, we can rewrite 
$$\sum_{i \in \widetilde{U}_{1,t}\backslash I_X} (|R_{ij_1}^\varepsilon| - |R_{ij_2}^\varepsilon| ) = \sum_{r=n_X(t)+1}^{p-p_0} (|R_{u_r j_1}^\varepsilon| - |R_{u_rj_2}^\varepsilon| )$$

~~
 
 By Assumption \ref{ass:ass2}, we have
\begin{eqnarray*}
    E\left[\left.e^{a(|R_{u_rj_1}^\varepsilon| - |R_{u_rj_2}^\varepsilon|) }\right| \mathcal{F}_{r-1} \right] \leq 1+4e^{-\frac{n\varepsilon^2}{2(s_4^2+s_5^2)}+|a|},
\end{eqnarray*}
where $\mathcal{F}_{r-1} =\sigma\{ \{\textbf{r.sum}_{\tau} \}_{\tau=1}^{t_X(r-1) }\}$ and $t_X(r)=\text{argmax}_{t'}\{n_X(t')=r\}$. Note that $\{|R_{u_{r'}j_1}^\varepsilon|,|R_{u_{r'}j_2}^\varepsilon|\}_{r'=1}^{r-1} \subset \mathcal{F}_{r-1}$. Letting $a=n\varepsilon^2/2(s_1^2+s_2^2)$, we have
\begin{eqnarray}
    E\left[\left.e^{\frac{n\varepsilon^2}{2(s_1^2+s_2^2)}(|R_{u_rj_1}^\varepsilon| - |R_{u_rj_2}^\varepsilon|) }\right| \mathcal{F}_{r-1}\right]  \leq  5.\label{eq:lem3step1}
\end{eqnarray}

For $i_1,i_2\in [p]$, let 
\begin{eqnarray*}
M_r^{i_1i_2}= \exp\left(\frac{n\varepsilon^2}{2(s_1^2+s_2^2)}\sum_{r'=1}^r(|R_{u_{r'}j_1}^\varepsilon| - |R_{u_{r'}j_2}^\varepsilon|)-r\log 5\right)    
\end{eqnarray*}
    where $M_r^{i_1i_2} = M_{p-p_0}^{i_1i_2}$ for $r\geq p-p_0$ and $r_\tau$ be a stopping time with respect to the filtration $\{\mathcal{F}_{r-1}\}$. First, we claim $\{M_r^{i_1i_2}\}$ is a supermartingale. Let $D_r^{i_1i_2} = \exp\left(\frac{n\varepsilon^2}{2(s_1^2+s_2^2)}(|R_{u_rj_1}^\varepsilon| - |R_{u_rj_2}^\varepsilon|)\right)$. 
By \eqref{eq:lem3step1}, we have
\begin{eqnarray*}
E[D_r^{i_1i_2}|\mathcal{F}_{r-1}] = E\left[\left.e^{\frac{n\varepsilon^2}{2(s_1^2+s_2^2)}(|R_{u_rj_1}^\varepsilon|-|R_{u_rj_2}^\varepsilon|) }\right|\mathcal{F}_{r-1}\right]5^{-1}\leq 1.
\end{eqnarray*}
Clearly, $D_r^{i_1i_2}$ is $\mathcal{F}_r^{i_1i_2}$-measurable, as is $M_r^{i_1i_2}$. Further, we have $E[M_r^{i_1i_2}|\mathcal{F}_{r-1}]=D_1^{i_1i_2}\cdots D_{t-1}^{i_1i_2}E[D_r^{i_1i_2}|\mathcal{F}_{r-1}] \leq M_{t-1}^{i_1i_2}$. This shows that $\{M_r^{i_1i_2}\}$ is a supermartingale. By the convergence theorem for nonnegative supermartingales \citep{doob1953stochastic}, $M_\infty^{i_1i_2}=\lim_{t\rightarrow \infty} M_r^{i_1i_2}$ is almost surely well-defined. Hence, $M_{\tau_r}^{i_1i_2}$ is well-defined. Next, let $Q_r^{i_1i_2} = M_{\min(\tau_r,r)}^{i_1i_2}$ be a stopped version of $\{M_r^{i_1i_2}\}$. By Fatou's lemma \citep{rudin1976principles}, we have
\begin{eqnarray*}
    E[\text{liminf}_{r\rightarrow \infty} Q_r^{i_1i_2}] \leq \text{liminf}_{r\rightarrow \infty}E[Q_r^{i_1i_2}] \leq 1.
\end{eqnarray*}
This shows that $E[M_\tau^{i_1i_2}]\leq 1$ holds. Lastly, from $E[M_\tau^{i_1i_2}]\leq 1$, for $0\leq \delta < 1$, we get
\begin{flalign*}
   &P\left(\frac{n\varepsilon^2}{2(s_1^2+s_2^2)}\sum_{r=1}^{\tau_r} (|R_{u_rj_1}^\varepsilon|- |R_{u_rj_2}^\varepsilon|) -\tau \log 5 > \log \delta^{-1}\right)
    =P\left(M_{\tau_r}^{i_1i_2}\delta^{-1} > 1 \right)\\ &\leq E[M_{\tau_r}^{i_1i_2} \delta] \leq \delta.
\end{flalign*}
In other words, we have
\begin{eqnarray*}
    \frac{n\varepsilon^2}{2(s_1^2+s_2^2)}\sum_{r=1}^{\tau_r} (|R_{u_rj_1}^\varepsilon| - |R_{u_rj_2}^\varepsilon|) -\tau_r  \log 5 > \log \delta^{-1}
\end{eqnarray*}
with probability at least $1-\delta$. Now, we aim at the calculation of the least probability such that
\begin{eqnarray*}
    \sum_{r=1}^{\tau_r} (|R_{u_rj_1}^\varepsilon| - |R_{u_rj_2}^\varepsilon|)   < \frac{p_0(\eta-2)\rho}{2\eta}
\end{eqnarray*}
for all $j_1,j_2\in [q]$. Even if we replace the stopping time $\tau_r$ with $r'$, 
\begin{eqnarray*}
    \sum_{r=1}^{r'} (|R_{u_rj_1}^\varepsilon| - |R_{u_rj_2}^\varepsilon|)   < \frac{p_0(\eta-2)\rho}{2\eta}
\end{eqnarray*}
is still satisfied with the same probability. For fixed $j_1,j_2$, the least probability satisfying the above can be calculated based on the inequality
\begin{eqnarray*}
      \frac{2(s_1^2+s_2^2)}{n\varepsilon^2} \left(\log \delta^{-1} + p\log 5 \right) \leq \frac{p_0(\eta-2)\rho}{2\eta},
\end{eqnarray*}
which leads to
\begin{eqnarray*}
    \delta = \exp\left(-\frac{p_0(\eta-2)\rho}{2\eta}\frac{n\varepsilon^2}{2(s_1^2+s_2^2)} + p\log 5\right).
\end{eqnarray*}

Accordingly, by applying this for all $j_1,j_2 \in [q]$,
we have
\begin{eqnarray*}
    \left|\sum_{r=1}^{r'} (|R_{u_rj_1}^\varepsilon| - |R_{u_rj_2}^\varepsilon|) \right|  < \frac{p_0(\eta-2)\rho}{2\eta},
\end{eqnarray*}
with probability at least $1-q^2\exp\left(-\frac{p_0(\eta-2)\rho}{2\eta}\frac{n\varepsilon^2}{2(s_1^2+s_2^2)} + p\log 5\right)$.

By letting $r' = p-p_0$, we have $\left|\sum_{r=1}^{p-p_0} |R_{u_rj_1}^\varepsilon| - |R_{u_rj_2}^\varepsilon|\right| \leq \frac{p_0(\eta-2)\rho}{2\eta}$. Accordingly, we have
$$\left|\sum_{r=r'}^{p-p_0} (|R_{u_rj_2}^\varepsilon| - |R_{u_rj_2}^\varepsilon|) \right| \leq \left|\sum_{r=1}^{p-p_0} |R_{u_rj_1}^\varepsilon| - |R_{u_rj_2}^\varepsilon|\right| + \left|\sum_{r=1}^{r'-1} |R_{u_rj_1}^\varepsilon| - |R_{u_rj_2}^\varepsilon|\right| \leq \frac{p_0(\eta-2)\rho}{\eta}.$$

Thus, we have

$$\left|\sum_{i \in \widetilde{U}_{1,t}\backslash I_X} (|R_{ij_1}^\varepsilon| - |R_{ij_2}^\varepsilon| ) \right| \leq \frac{p_0(\eta-2)\rho}{\eta}.$$
for all $t=1,2,\dots,p+q-1$, with probability $1-\delta$. Similarly, we have

$$\left|\sum_{j \in \widetilde{V}_{1,t}\backslash I_Y} (|R_{i_1j}^\varepsilon| - |R_{i_2j}^\varepsilon| )\right| \leq \frac{q_0(\eta-2)\rho}{\eta}.$$
for all $i_1,i_2\in[p]$ and all $t=1,2,\dots,p+q-1$, with probability $1-p^2\exp\left(-\frac{q_0(\eta-2)\rho}{2\eta}\frac{n\varepsilon^2}{2(s_1^2+s_2^2)} + q\log 5\right)$. Putting these two together, with probability at least $1-q^2\exp\left(-\frac{p_0(\eta-2)\rho}{2\eta}\frac{n\varepsilon^2}{2(s_1^2+s_2^2)} + p\log 5\right)-p^2\exp\left(-\frac{q_0(\eta-2)\rho}{2\eta}\frac{n\varepsilon^2}{2(s_1^2+s_2^2)} + q\log 5\right)$, we have 
$$\left|\sum_{i \in \widetilde{U}_{1,t}\backslash I_X} (|R_{ij_1}^\varepsilon| - |R_{ij_2}^\varepsilon| ) \right| \leq \frac{p_0(\eta-2)\rho}{\eta}$$
and
$$\left|\sum_{j \in \widetilde{V}_{1,t}\backslash I_Y} (|R_{i_1j}^\varepsilon| - |R_{i_2j}^\varepsilon| )\right| \leq \frac{q_0(\eta-2)\rho}{\eta}.$$

\end{proof}

Now, we are ready to prove Theorem \ref{thm:1}.

\begin{proof}

At the round $t$ for $c=1$, the greedy algorithm considers the row sums $\sum_{j \in  \widetilde{V}_{1,t}} |R_{ij}^\varepsilon|$ for all $i\in \widetilde{U}_{1,t}$ and column sums $\sum_{i \in  \widetilde{U}_{1,t}} |R_{ij}^\varepsilon|$ for all $j\in \widetilde{V}_{1,t}$. We will show that the greedy algorithm excludes all rows in $U\backslash I_X$ before excluding ones in $I_X$. Then, we show that the same logic works for the columns. It suffices to show that $\sum_{j \in  \widetilde{V}_{1,t}} |R_{i_1j}^\varepsilon| \geq \sum_{j \in  \widetilde{V}_{1,t}} |R_{i_2j}^\varepsilon|$ for all $i_1 \in I_X$ and $i_2 \in I_X$ and for all $t$. Note that $\sum_{j \in  \widetilde{V}_{1,t}} |R_{i_1j}^\varepsilon| - \sum_{j \in  \widetilde{V}_{1,t}} |R_{i_2j}^\varepsilon|\geq 0$ can be written as 
\begin{eqnarray}
    \sum_{j \in  \widetilde{V}_{1,t}\cap I_Y} (|R_{i_1j}^\varepsilon| -  |R_{i_2j}^\varepsilon|) + \sum_{j \in  \widetilde{V}_{1,t}\backslash I_Y} (|R_{i_1j}^\varepsilon| -  |R_{i_2j}^\varepsilon|)\label{eq:decomp}  
\end{eqnarray}

By Lemma \ref{lem:3}, we have
 \begin{eqnarray*}
     \sum_{j \in  I_Y} (|R_{i_1j}^\varepsilon| -  |R_{i_2j}^\varepsilon|) \geq \frac{q_0(\eta-2)\rho}{\eta},
 \end{eqnarray*}
for $i_1 \in I_X$ and $i_2 \in U\backslash I_X$ .
At rount $t=1$, \eqref{eq:decomp} is $\sum_{j \in  I_Y} (|R_{i_1j}^\varepsilon| -  |R_{i_2j}^\varepsilon|) + \sum_{j \in  V\backslash I_Y} (|R_{i_1j}^\varepsilon| -  |R_{i_2j}^\varepsilon|)$.  Since  
$$\left|\sum_{j \in \widetilde{V}_{1,t}\backslash I_Y} (|R_{i_1j}^\varepsilon| - |R_{i_2j}^\varepsilon| ) \right|\leq \frac{q_0(\eta-2)\rho}{\eta},$$
for all $i_1,i_2$ and all $t=1,2,\dots,p+q-1$, with probability $1-\delta_1$, by Lemma \ref{lem:4}, we have $\sum_{j \in \widetilde{V}_{1,t}} (|R_{i_1j}^\varepsilon| - |R_{i_2j}^\varepsilon| ) > 0$ for all $i_1 \in I_X$ and $i_2 \in U\backslash I_X$, leading to $\text{argmin}_i(\sum_{j \in \widetilde{V}_{1,1}} |R_{ij}^\varepsilon|) \in U\backslash I_X$. For the same reason, we have $\text{argmin}_j(\sum_{i \in \widetilde{U}_{1,1}} |R_{ij}^\varepsilon|) \in V\backslash I_Y$. Thus, either a row in $U\backslash I_X$ or a column in $V\backslash I_Y$ is excluded. For the subsequent rounds, this same argument applies: either a row in $U\backslash I_X$ or a column in $V\backslash I_Y$ is excluded until only rows in $I_X$ and columns in $I_Y$ are left as $\sum_{j \in \widetilde{V}_{1,t}} (|R_{i_1j}^\varepsilon| - |R_{i_2j}^\varepsilon| ) > 0$ $i_1 \in I_X$ and $i_2 \in U\backslash I_X$ and $\sum_{i \in \widetilde{U}_{1,t}} (|R_{ij_1}^\varepsilon| - |R_{ij_2}^\varepsilon| ) > 0$ for $j_1 \in I_Y$ and $j_2 \in V\backslash I_Y$. Therefore, if $|I_X| =p_0\leq p_{phase1}$ and $|I_Y| =q_0 \leq q_{phase1}$, the true submodel $\mathcal{M}_{\star}$ is included in the submodel created by the intermediate screened variables $\widehat{\mathcal{M}}_{p1}$ with probability at least $1-\delta(p,p_0,q,q_0,n)$, where
\begin{eqnarray*}
\delta(p,p_0,q,q_0,n)&=&\exp \left(-n\rho^2/2\eta^2 s_1^2+\log 2pq \right)\\
&+&\exp\left(-\frac{p_0(\eta-2)\rho}{2\eta}\frac{n\varepsilon^2}{2(s_1^2+s_2^2)} + p\log 5 + 2\log q \right)\\
&+&\exp\left(-\frac{q_0(\eta-2)\rho}{2\eta}\frac{n\varepsilon^2}{2(s_1^2+s_2^2)} + q\log 5 + 2\log p \right).
\end{eqnarray*}
Because the last two terms are the leading terms, by choosing $\eta$ minimizing the sum of terms, we can have $c_1,~c_2>0$ such that
\begin{eqnarray*}
\delta(p,p_0,q,q_0,n)=O\left(\exp\left(-c_1p_0n + c_2p \right)+\exp\left(-c_1q_0n + c_2q \right)\right).
\end{eqnarray*}
\end{proof}

\subsection{Proof of Theorem \ref{thm:2}}
\begin{proof}
Now, we show that there is a set of $\lambda$ such that the objective function $\text{den}_\lambda$ is maximized at the time when only and all the associated rows and columns remain in our active node set. To proceed, we set $w\geq 1$ as a constant $w=\max_{e_{ij}\in E_1} |\rho_{ij}|/\rho$. First, we compare two cases: (i) only all rows $I_X$ and columns in $I_Y$ remain in our active node set, (ii) $r_1$ rows in $U\backslash I_X$ and $c_1$ columns $V\backslash I_Y$, the index sets of which are denoted by $I_X''$ and $I_Y''$,  remain in our active node set including all $i \in I_X$ and $j\in I_Y$ at time $t$. Then, the values of the objective functions for case (i) and (ii) are as follows:
\begin{eqnarray*}
\text{den}_\lambda(B_{case_1})=\frac{\sum_{i\in I_X,j\in I_Y} |R_{ij}^\varepsilon|}{(p_0 q_0)^\lambda}
\end{eqnarray*}
and
\begin{flalign*}
\text{den}_\lambda(B_{case_2})=\frac{\sum_{i\in I_X,j\in I_Y} |R_{ij}^\varepsilon| + \sum_{i\in I_X,j\in I_Y''} |R_{ij}^\varepsilon| +\sum_{i\in I_X'',j\in I_Y} |R_{ij}^\varepsilon| + \sum_{i\in I_X'',j\in I_Y''} |R_{ij}^\varepsilon|}{( (p_0+r_1)(q_0+c_1) )^\lambda}
\end{flalign*}

Applying $|R_{ij}^\varepsilon| > ((\eta-1)/\eta)\rho$ for $(i,j) \in I_X\otimes I_Y$ and $|R_{ij}^\varepsilon| < (1/\eta)\rho$ otherwise, we derive the following inequality:
\begin{eqnarray}
\lambda > \log \left(\frac{\eta-2}{\eta-1} + \frac{1}{\eta-1}\frac{(p_0+r_1)(q_0+c_1)}{p_0q_0}\right)/\log \left(\frac{(p_0+r_1)(q_0+c_1)}{p_0q_0}\right)    \label{eq:lambda1}
\end{eqnarray}

We consider case (iii) such that only some $r_2$ rows in $I_X$ and $c_2$ columns in $I_Y$ remain in our active node sets, which are denoted by $J_X'$ and $J_Y'$. Then, we have
\begin{eqnarray*}
    \text{den}_\lambda(B_{case_3})= \frac{\sum_{i\in J_X',j\in J_Y'} |R_{ij}^\varepsilon|}{( r_2c_2 )^\lambda}.
\end{eqnarray*}

Now, we find the range of $\lambda$ such that the objective function of case (i) is greater than that of case (iii). By applying the upper and lower bounds of $|R_{ij}^\varepsilon| \leq ((\eta+w)/\eta)\rho$ for $(i,j) \in J_X'\otimes J_Y'$ and $|R_{ij}^\varepsilon| \geq ((\eta-1)/\eta)\rho$ for $(i,j) \in (I_X\otimes I_Y)\backslash (J_X'\otimes J_Y')$, respectively, for the inequality $\frac{\sum_{i\in I_X,j\in I_Y} |R_{ij}^\varepsilon|}{(p_0 q_0)^\lambda} \geq  \frac{\sum_{i\in J_X',j\in J_Y'} |R_{ij}^\varepsilon|}{( r_2c_2)^\lambda}$, we have
\begin{eqnarray}
\lambda < \log \left(\frac{w+1}{\eta+w} + \frac{\eta-1}{\eta+w}\frac{p_0q_0}{r_2c_2}\right)/\log \left(\frac{p_0q_0}{r_2c_2}\right).\label{eq:lambda2}
\end{eqnarray}
Putting \eqref{eq:lambda1} and \eqref{eq:lambda2} together, we have
\begin{eqnarray*}
    0<\frac{\log \left(\frac{\eta-2}{\eta-1} + \frac{1}{\eta-1}\frac{(p_0+r_1)(q_0+c_1)}{p_0q_0}\right)}{\log \left(\frac{(p_0+r_1)(q_0+c_1)}{p_0q_0}\right)} <\lambda< \frac{\log \left(\frac{w+1}{\eta+w} + \frac{\eta-1}{\eta+w}\frac{p_0q_0}{r_2c_2}\right)}{\log \left(\frac{p_0q_0}{r_2c_2}\right)} < 1.
\end{eqnarray*}
To satisfy the above inequality for all $1\leq r_1\leq p-p_0$, $1\leq c_1< q-q_0$, $0<r_2<p_0$, and $0<c_2<q_0$, we need sufficiently large $\eta$. If we apply the ranges of $r_1,~c_1,~r_2$ and $c_2$, we can show that the inequality can be written as
\begin{eqnarray*}
   \log\left(\frac{\eta-2}{\eta-1} +\frac{1}{\eta-1} a_1\right)/\log a_1 < \log\left( \frac{w+1}{\eta+1} +\frac{\eta-1}{\eta+w}a_2 \right)/\log a_2,
\end{eqnarray*}
where $a_1 = \min(\frac{p_1+1}{p_0},\frac{q_0+1}{q_0})>1$ and $a_2 = \min(\frac{p_0}{p_0-1},\frac{q_0}{q_0-1})>1$. As $a_1<a_2$ and the term on the RHS above is increasing in $a_2 > 1$, we can replace $a_2$ with $a_1$. Then, we get $\log\left(\frac{\eta-2}{\eta-1} +\frac{1}{\eta-1} a_1\right) < \log\left( \frac{w+1}{\eta+w} +\frac{\eta-1}{\eta+w}a_1\right)$. By simple calculation using $a_1>1$, we have $ \eta^2 - 3\eta +(1-w) > 0$. Using the quadratic formula, we have 
\begin{eqnarray}
\eta > (3+\sqrt{9+4(w-1)})/2\geq 3,  \label{eq:quad}
\end{eqnarray}
where $w\geq 1$. Thus, we show that there exists $\lambda$ satisfying that the objective function is maximized at the time when only all rows in $I_X$ and columns $I_Y$ are in our active node set.

We choose $\eta_\star$ minimizing $\delta_0+\delta_1$, which is not analytically solvable, subject to \eqref{eq:quad}. Let the confidence level be
\begin{eqnarray*}
\delta(p,p_0,q,q_0,n)&=&\exp \left(-n\rho^2/2\eta_\star^2 s_1^2+\log 2pq \right)\\
&+&\exp\left(-\frac{p_0(\eta_\star-2)\rho}{2\eta_\star}\frac{n\varepsilon^2}{2(s_1^2+s_2^2)} + p\log 5+ 2\log q \right)\\
&+&\exp\left(-\frac{q_0(\eta_\star-2)\rho}{2\eta_\star}\frac{n\varepsilon^2}{2(s_1^2+s_2^2)} + q\log 5 + 2\log p \right).
\end{eqnarray*}

Here, the last two leading terms can be written as
\begin{eqnarray*}
\exp\left(-c_1'p_0n + c_2'p +2\log q\right)+\exp\left(-c_1'q_0n + c_2'q + 2\log p  \right),
\end{eqnarray*}
for some $c_1',~c_2'>0$. Thus, $\delta(p,p_0,q,q_0,n) = O(\exp\left(-p_0c_1' n + c_2'p \right)+\exp\left(-q_0c_1'n + c_2'q  \right))$. This leads to the minimum sample condition $n = O(\max\{p/p_0,q/q_0\})$ for the exact recovery.

\end{proof}

\label{lastpage}

\end{document}